\documentclass[margin=5in, 12pt]{article}

\usepackage[utf8]{inputenc}
\usepackage[affil-it]{authblk}
\usepackage[margin=1in]{geometry}
\usepackage{multirow}
\usepackage{natbib}
\usepackage{graphicx}
\usepackage{float}
\graphicspath{{figures/}{./}}
\usepackage{amsmath}
\usepackage{amssymb}
\usepackage{amsfonts}
\usepackage{amstext}
\usepackage{bbm}
\usepackage{subcaption}
\usepackage[dvipsnames]{xcolor}
\usepackage{newpxtext,newpxmath}
\usepackage{url}
\usepackage{bm}
\usepackage{amsthm}

\newtheorem{remark}{Remark}

\newtheorem{lemma}{Lemma}
\newtheorem{proposition}{Proposition}
\usepackage{booktabs}

\DeclareMathOperator{\E}{\mathbb{E}}
\DeclareMathOperator{\Var}{Var}

\newcommand{\Dtd}{\tilde{D}}

\title{\textbf{Partial Homogeneity in Staggered Difference-in-Differences}%
       }
\author{Parush Arora$^*$ and Rohan Wagle\thanks{Department of Economics, Ashoka University.}}
\date{}

\begin{document}
\maketitle

\begin{abstract}
  \noindent
  In staggered difference-in-differences (DiD) designs, units enter treatment at
  different calendar times, so the treatment effect is not a single number but a
  set of Cohort-Average Treatment effects on the Treated (CATTs), one per
  cohort-time cell. Estimating every CATT as its own parameter, as the standard
  fully flexible estimator does, is unbiased but inefficient when some of these
  effects are in fact equal, whereas pooling them all into a single two-way fixed effects
  (TWFE) coefficient is efficient but, whenever the heterogeneity is genuine, biased for the individual effects.
  We frame the choice between these extremes as a partition-selection
  problem on the cohort-time cells and address it with a Dirichlet
  Process (DP) mixture prior on the CATTs. The model favors parsimonious groupings
  without fixing their number, and a collapsed Gibbs sampler delivers point
  estimates, credible intervals that marginalize the unknown partition, and
  co-clustering probabilities for every pair of CATTs. With the error variance
  held fixed and a pairwise penalty placed on the partition, a maximum a
  posteriori (MAP) partition reduces to an
  $\ell_0$-penalized regression, connecting the Bayesian formulation to the
  homogeneity-pursuit literature.
  In a calibrated simulation, the model cuts the sampling variance of the
  cohort-time effects by 26--52\% relative to the fully flexible estimator, without
  the pooled estimator's bias, provided the distinct effects are separated enough
  to be recovered, and the posterior delivers near-nominal confidence-interval
  coverage by averaging over the unknown partition.
  In two applications the method recovers a precision-improving
  partial-homogeneity structure in one, where the cohort-time effects are
  genuinely heterogeneous, and reports that full pooling is adequate in the
  other, where they are not.
\end{abstract}



\newpage

\section{Introduction}
\label{sec:intro}

Difference-in-differences (DiD) has become one of the workhorses of empirical economics. In its canonical form, a treatment cohort and a control cohort are observed before and after a single policy change, and the researcher estimates a single average treatment effect (ATT or average treatment effect on treated). In a simple setting with two time periods, ATT is recovered by taking the difference between the change in the average outcome over time for the treatment and control cohorts, commonly referred to as a simple $2\times2$ DiD estimate. In a general setting with $N$ observations and $T$ time periods, the most popular way to estimate the ATT with a binary treatment indicator is the two-way fixed effects regression (TWFE):
\begin{equation}
\label{eq:rest_twfe}
  Y_{it} = \alpha_i + \lambda_t + \tau D_{it} + \varepsilon_{it},
\end{equation}
where $Y_{it}$ is the observed outcome for $i^{th}$ unit in time $t$, $\alpha_i$ is the unit fixed effect, $\lambda_t$ is the time fixed effect, $D_{it}$ is the treatment dummy and $\tau$ is the ATT. We will refer to Eq.~\ref{eq:rest_twfe} as pooled TWFE. The pooled TWFE estimator ($\hat\tau_{\text{pool}}$) is directly related to the $2\times 2$ DiD estimator, as it can be represented as a weighted average of all possible 2$\times$2 DiD estimates \citep{goodman2021}.

Modern policy evaluations, however, rarely fit this clean template. Minimum wage increases, Medicaid expansions, trade liberalizations, and school reforms typically roll out across states, cities, or firms in a staggered fashion, with different units adopting treatment at different calendar times. We update the notation so that $Y_{igt}$ and $D_{igt}$ denote the outcome and treatment indicator for unit $i$ in cohort $g$ at time $t$. The estimator $\hat\tau_{\text{pool}}$ is unbiased in staggered timing if the treatment effect for each $g\times t$ case (cohort-time specific ATT or CATT) is homogeneous.

The econometrics of staggered DiD has seen remarkable progress in recent years. \citet{dechaisemartin2020}, \citet{callaway2021}, \citet{sun2021}, \citet{goodman2021}, \citet{gardner2022two}, and
\citet{borusyak2024}, among others, have shown that the TWFE regression can produce a biased estimate of the ATT when there is heterogeneity among CATTs. \citet{dechaisemartin2023} and \citet{roth2023} survey this literature. The bias arises due to the inclusion of an already-treated cohort as a control, which violates the parallel trends assumption. \citet{goodman2021} showed that the decomposition of TWFE can contain negative weights when the treatment effects are heterogeneous. The recommended solution is to either drop the already-treated cohorts from the control or estimate the CATTs one coefficient per (cohort-time) cell and then aggregate them as desired. For the latter, \citet{wooldridge2025} discussed the flexible TWFE that allows interaction between the treatment dummy and the cohort-time dummy and recovers unbiased CATTs. Their estimate is identical to \citet{gardner2022two} and \citet{borusyak2024} in certain settings and more efficient than \citet{callaway2021} and \citet{sun2021}, which drop the already-treated cohorts from the control. These heterogeneity-robust estimators are efficient only under spherical errors. \citet{arora2026} show that efficiency under arbitrary within-unit serial correlation is attainable through an optimally weighted GMM in the space of clean comparisons.

This fully flexible approach solves the bias problem but introduces a new challenge, namely model specification. With $G$ treated cohorts and $T$ time
periods, there can be $K = \sum_{g=2}^{T} (T - g + 1)$ distinct CATT cells, where $g=2,3,\dotsc,T$ indexes the treated cohorts by the period in which each is first treated and $g = 0$ denotes the never-treated cohort, which contributes no post-treatment cells. Estimating all $K$ separately is unbiased, but may be inefficient if some CATTs share the same true effect.
Conversely, pooling all CATTs into a single parameter (classical TWFE) is
efficient but biased when effects genuinely differ across cohorts or periods.

The specification question is which CATTs can be pooled, and which
must be kept separate? This is a model selection problem on the space of
partitions of the $K$ CATT cells. Getting the partition right matters, because under
the true partition, the estimator is both unbiased and more efficient than the
fully flexible approach, exploiting the additional identifying restriction that
some effects are equal without imposing false restrictions on others.

Although the recent staggered-DiD literature has focused overwhelmingly on removing the bias of pooled TWFE, the complementary question of how to efficiently pool the resulting cohort-time effects has received comparatively little attention. The problem of grouping otherwise distinct coefficients that share a common value is, however, well studied in the broader panel-data and statistics literatures, and our approach draws on established machinery from both traditions. On the penalized-regression side, penalizing pairwise differences to recover latent groups underlies the classifier-Lasso of \citet{su2016} and the homogeneity-pursuit estimators of \citet{ke2015} and \citet{wang2018}, as well as the grouping-pursuit surface of \citet{shen2010}. Related ideas appear in the grouped fixed-effects approach of \citet{bonhomme2015} and in work on determining the number of latent groups \citep{lu2017}. Relative to these $\ell_1$-type fusion penalties, we use an $\ell_0$ penalty because it produces exact equality of coefficients rather than continuous shrinkage. On the Bayesian side, clustering coefficients through a Dirichlet Process mixture or an explicit partition prior is standard in Bayesian nonparametrics \citep{ferguson1973, antoniak1974, hartigan1990, barry1992}, and Bayesian nonparametric methods have also been applied directly to causal-effect estimation \citep{hill2011}. Our contribution is to bring these tools to bear on the specification problem in staggered DiD, to characterize the resulting bias-variance and identification trade-offs in that setting, and to connect the two traditions, since the $\ell_0$ fusion penalty is the fixed-variance maximum a posteriori estimator of a partition model in the same family as our Dirichlet Process model, differing from it only in the partition prior.

Our restriction is also related to, but distinct from, recent work that bounds the magnitude of treatment-effect heterogeneity rather than its structure. \citet{kwon2026} restrict the variance of conditional average treatment effects and conduct bias-aware sensitivity analysis over that bound, shrinking heterogeneous effects smoothly toward a common value; \citet{armstrong2025} adapt to misspecification of a restricted model, and \citet{armstrong2018} give the underlying optimal-inference theory; in a similar honest-inference spirit but for a different DiD assumption, \citet{rambachan2023} allow bounded violations of parallel trends. Where these papers impose a continuous bound and never assert that any two effects are exactly equal, we impose a discrete partial-homogeneity restriction and attempt to recover the grouping from the data. The two views are complementary, in that the continuous bound is agnostic about which effects coincide, whereas the discrete view is informative about the grouping structure itself, at the cost of a stronger assumption whose selection uncertainty must be confronted (a point we return to below and in the simulations).

This paper makes three contributions. First, we formalize the specification problem in staggered DiD as a partition-selection problem and characterize the bias and variance of the fully flexible, partially homogeneous, and fully pooled estimators under a partial-homogeneity data-generating process. Second, we propose a Bayesian solution, a Dirichlet Process (DP) mixture prior over the CATTs that assigns positive probability to every partition while favoring parsimonious groupings. A collapsed Gibbs sampler estimates the model, and the posterior delivers point estimates, credible intervals that marginalize the unknown partition, and co-clustering probabilities for every pair of CATTs. Third, we record a connection to the penalized-regression literature. With the error variance fixed and a flat prior on the group effects, the maximum a posteriori (MAP) partition under a pairwise partition prior (a proper prior in the same family as, but distinct from, the Dirichlet Process) is exactly an $\ell_0$-penalized least-squares estimator whose penalty equals that log prior. Under the Dirichlet Process prior itself the collapsed posterior instead carries determinant and shrinkage terms, and integrating the variance out gives a BIC-type criterion via a Schwarz approximation. This links the Bayesian formulation to $\ell_0$ homogeneity pursuit.

The DP prior is the standard nonparametric choice for models in which the number of components is unknown \citep{ferguson1973, antoniak1974}. \citet{escobar1995} show how DP mixture models can be used for density estimation with an unknown number of mixture components, and \citet{muller1996} extend the approach to regression. \citet{neal2000} develops efficient MCMC algorithms for posterior inference,
including the collapsed Gibbs sampler that allows for efficient drawing from the posterior (which we employ). \citet{muller2004} provide
a survey of nonparametric Bayesian methods for applied researchers. \citet{miller2018} study the finite mixture model with a prior on the number of components, which shares the spirit of our approach.

The remainder of the paper is organized as follows. Section~\ref{sec:problem} formalizes the specification problem, establishes the bias-variance trade-off under partial homogeneity.
Section~\ref{sec:bayes} develops the Dirichlet Process model, covering its specification, estimation, and inference, and records the $\ell_0$-penalized estimator that arises as a fixed-variance MAP within the same partition-model family.
Section~\ref{sec:sim} reports the simulation results, Section~\ref{sec:empirical_application} presents two empirical applications, and Section~\ref{sec:conc} concludes.
\section{The Framework}
\label{sec:problem}

\subsection{Setup and Assumptions}
Consider a balanced panel with $N$ units observed over $T$ time periods. We are interested in estimating the effect of a binary treatment, $D_{it} \in \{0,1\}$, on the outcome variable, $Y_{igt} \in \mathbb{R}$. Let $Y_{igt}(1)$ and $Y_{igt}(0)$ denote the potential outcomes with and without treatment, related to the observed outcome by $Y_{igt} = D_{it}\,Y_{igt}(1) + (1-D_{it})\,Y_{igt}(0)$. Assume we observe the sample $\{Y_{ig1},$ $\dotsc,Y_{igT}$ $,D_{ig1},$ $\dotsc,D_{igT}\}$ where $i = 1,\dotsc, N$ indexes units and $t = 1,\dotsc, T$ indexes time. Units belong to one of $g$ treated cohorts (indexed by treatment entry date $g \in \{2,3,\dotsc,T\} = \mathcal{G}$) or a never-treated group ($g=0$). Period 1 is assumed to be pre-treatment for all cohorts, which means $g\neq 1$. Let $N_g$ be the number of units in the cohort $g$ and thus $\sum_{g}N_g = N$. We make the following assumptions about the treatment process:

\begin{itemize}
    \item[] \textbf{Assumption 1} (Random Sample): $\{Y_{ig1},\dotsc,Y_{igT}$ $,D_{i1},\dotsc,D_{iT}\}_{i=1}^{N}$ are independent and identically distributed. Each sampled unit contributes its complete length-$T$ trajectory, consistent with the balanced panel maintained throughout.
    \item[] \textbf{Assumption 2} (Irreversibility of Treatment): It implies that once a unit is treated, that unit remains treated in the next period. Units do not ``forget'' about the treatment experience. If $D_{i,t-1} = 1$ then $D_{i,t} = 1$.
    \item[] \textbf{Assumption 3} (Parallel Trends): In the absence of treatment, the outcome follows the same trend across cohorts. Formally, for every pair of adjacent periods $(t-1,t)$, the expected change in the untreated potential outcome, $E\bigl[Y_{igt}(0) - Y_{ig,t-1}(0)\bigr]$, does not depend on the cohort $g$ (with $g=0$ the never-treated group). Because it is stated on the untreated potential outcome, the restriction is well defined in every period, including post-adoption periods $t \geq g$ where $Y_{igt}(0)$ is counterfactual, and it is unaffected by which cohorts have already adopted. The assumption may hold either unconditionally or only after conditioning on covariates; we allow the conditional version, that $E\bigl[Y_{igt}(0) - Y_{ig,t-1}(0)\mid X_{igt}\bigr]$ does not depend on $g$, which accommodates (1) covariate-specific trends in $Y_{igt}(0)$ and (2) different distributions of covariates across cohorts, provided the compared cohorts share common support in $X_{igt}$, without requiring it. We impose it at the unit level, though our estimator remains valid under the cohort-level version.
    \item[] \textbf{Assumption 4} (No Treatment Anticipation): There is no treatment effect in pre-treatment periods. The outcome $Y_{igt}$ in any period before $g$ is, on average, equal to the untreated outcome.
\end{itemize}
\subsection{The Specification Problem}
Under homogeneous effects the CATTs all equal the ATT and $\hat\tau_{\text{pool}}$ is unbiased. Under heterogeneity they differ, and the target ATT is a weighted average of the CATTs,
\begin{equation}
\label{eq:theta}
	\tau = \sum_{g}\sum_{t} w_{gt}\,\tau_{gt},
\end{equation}
where $\tau_{gt}$ is the CATT for cohort $g$ at calendar time $t$ and \citet{callaway2021} discusses the choice of weights. As reviewed in Section~\ref{sec:intro}, $\hat\tau_{\text{pool}}$ is then biased for $\tau$, because already-treated units enter as controls (the forbidden comparison) and some weights in the \citet{goodman2021} decomposition turn negative. The literature offers several unbiased alternatives, restricting the controls to never- or not-yet-treated units \citep{callaway2021, sun2021}, allowing treatment to switch on and off \citep{dechaisemartin2020}, or estimating one coefficient per cohort-time cell through the flexible and imputation estimators of \citet{wooldridge2025}, \citet{gardner2022two}, \citet{borusyak2024}, and \citet{liu2024practical}, which are unbiased and efficient under spherical errors, with \citet{arora2026} extending efficiency to arbitrary within-unit serial correlation. We build on the flexible TWFE of \citet{wooldridge2025}, which interacts the treatment dummy with the cohort-time dummies,
\begin{equation}
  Y_{igt} = \alpha_i + \lambda_t
    + \sum_{g \in \mathcal{G}} \sum_{t \geq g} \tau_{gt}\, D_{igt}
    + \varepsilon_{it},
  \label{eq:twfe_flex}
\end{equation}
and recovers unbiased CATTs $\hat\tau_{gt}$, from which the ATT is estimated as $\hat\tau_{\text{flex}} = \sum_{g}\sum_{t} w_{gt}\,\hat\tau_{gt}$.

Let $Y$ be the $NT \times 1$ vector of outcomes and $D = [D_1, \ldots, D_K]$ be the $NT \times K$ matrix of cohort-time dummies. Let $\tilde{Y} = Y - \bar{Y}_i - \bar{Y}_t + \bar{Y}$ be the vector of within-transformed outcomes and $\tilde{D} = [\tilde{D}_1, \ldots, \tilde{D}_K]$ be the matrix of within-transformed cohort-time dummies where $\tilde{D}_k = D_k - \bar{D}_{ki} - \bar{D}_{kt} + \bar{D}_k$. Here, $\bar Y_i$ and $\bar D_{ki}$ are the unit-specific averages over time, $\bar Y_t$ and $\bar D_{kt}$ are the time-specific averages over units and $\bar Y$ and $\bar D_k$ are the overall averages. Double-demeaning is the orthogonal projection onto the residual space of the
unit and time fixed effects, a subspace of dimension $r = NT - N - T + 1$.
Representing that space by an orthonormal basis $Q$ with $Q'Q = I_r$, and reading
$\tilde{Y}$ and $\tilde{D}$ in these $r$ coordinates, the within-transformed model
is
\begin{equation}
  \tilde{Y} = \Dtd\,\tau + \tilde{\varepsilon},
  \qquad
  \tilde{\varepsilon} \sim \bigl(\mathbf{0},\,\sigma^2 I_r\bigr),
  \label{eq:flex_dm}
\end{equation}
where $\tau = (\tau_1,\ldots,\tau_K)'$ collects all $K$ CATTs. The transformation
enters every estimator below only through the cross-products $\Dtd'\Dtd$ and
$\Dtd'\tilde{Y}$, which coincide in the double-demeaned and orthonormal
representations, so we use the two interchangeably.
Ordinary least squares on \eqref{eq:flex_dm} yields the fully
flexible estimator
\begin{equation}
  \hat{\tau}_{\text{flex}}
  = \bigl(\Dtd'\Dtd\bigr)^{-1}\Dtd'\tilde{Y} = [\hat{\tau}_{1,\text{flex}},\dotsc,\hat{\tau}_{K,\text{flex}}]' \qquad
  \Var\bigl(\hat{\bm{\tau}}_{\text{flex}}\bigr)
  = \sigma^2\bigl(\Dtd'\Dtd\bigr)^{-1}.
  \label{eq:flex_var}
\end{equation}
The $k^{th}$ diagonal element of $\sigma^2[(\Dtd'\Dtd)^{-1}]$ is the sampling variance of the $k^{th}$ CATT estimate. Because each CATT is estimated from the variation in its own
within-transformed indicator, this variance is governed
by how much identifying variation CATT $k$ possesses individually.
In large panels with many cohorts and long post-treatment windows, $K$
is large and the variation attributable to any single CATT cell is
correspondingly modest, leading to imprecise estimates even when the
overall sample size $NT$ is large.

The opposite extreme imposes a single common coefficient $\tau$ for every
cohort-time cell.
Writing $\tilde{D}_{\text{pool}} = \sum_{k=1}^K \tilde{D}_k$ for the
aggregate within-transformed treatment indicator, the fully pooled
estimator is
\begin{equation}
  \hat\tau_{\text{pool}}
  = \bigl(\tilde{D}_{\text{pool}}'\tilde{D}_{\text{pool}}\bigr)^{-1}
    \tilde{D}_{\text{pool}}'\tilde Y,
  \qquad
  \Var\!\bigl(\hat\tau_{\text{pool}}\bigr)
  = \frac{\sigma^2}{\bigl\|\tilde{D}_{\text{pool}}\bigr\|^2}.
  \label{eq:pool_var}
\end{equation}
Since $\tilde{D}_{\text{pool}}$ aggregates the variation of all $K$ cells,
$\|\tilde{D}_{\text{pool}}\|^2 > \|\tilde{D}_k\|^2$ for any individual
$k$, and the pooled estimator is correspondingly precise.
The cost is bias. The pooled estimator is consistent for a weighted
average of the true CATTs, $\sum_{k=1}^K w_k \tau_k^*$, not for any individual $\tau_k^*$.
Formally, the bias of $\hat\tau_{\text{pool}}$ as an estimator of $\tau_k^*$ is
\begin{equation}
  \E\bigl[\hat\tau_{\text{pool}}\bigr] - \tau_k^*
  = \sum_{j=1}^K w_j\bigl(\tau_j^* - \tau_k^*\bigr),
  \qquad
  w_j = \frac{\tilde{D}_{\text{pool}}'\tilde{D}_j}{\tilde{D}_{\text{pool}}'\tilde{D}_{\text{pool}}},
  \label{eq:pool_bias}
\end{equation}
the implicit weights pooled TWFE places on the cohort-time effects, which sum to
one. When the within-transformed dummies are mutually orthogonal these reduce to
the positive norm weights $w_j = \|\tilde{D}_j\|^2/\sum_{\ell}\|\tilde{D}_\ell\|^2$;
in general the cross-products $\tilde{D}_{\text{pool}}'\tilde{D}_j$ left by the
shared fixed effects need not be positive, and $w_j$ then coincides with the
Goodman--Bacon weighting, which can turn negative where already-treated units
serve as controls.
The bias in \eqref{eq:pool_bias} for cell $k$ equals $\sum_{j} w_j\tau_j^* - \tau_k^*$, the gap between the weighted average and cell $k$; it is zero for every cell simultaneously only when all true CATTs are equal, the case in which the pooled model is correctly specified, though for a particular cell it can vanish when that cell's effect happens to coincide with the weighted average. Under genuine heterogeneity a single pooled coefficient cannot reproduce the full CATT vector, whatever weighted average it targets, so it is uninformative for heterogeneity-aware policy analysis even where it happens to be unbiased for an individual cell. The estimators we propose pool clean cohort-time effects rather than raw observations, so any grouping short of full pooling combines only clean comparisons; away from that corner the specification problem we study is one of aggregation among clean effects rather than the negative-weight problem of raw TWFE, with which the fully pooled benchmark coincides.

The flexible specification \eqref{eq:twfe_flex} is an ordinary linear regression, so covariate adjustment needs no new machinery. To accommodate the conditional version of Assumption~3, we augment it with covariates $X_{igt}$, time-varying controls or covariate-cohort interactions that permit covariate-specific trends,
\begin{equation}
  Y_{igt} = \alpha_i + \lambda_t + X_{igt}'\beta + \sum_{g \in \mathcal{G}} \sum_{t \geq g} \tau_{gt}\, D_{igt} + \varepsilon_{it}.
  \label{eq:twfe_cov}
\end{equation}
By the Frisch--Waugh--Lovell theorem the CATT estimates are identical whether $X$ enters the regression directly or is partialled out beforehand. Writing $M$ for the residual-maker that projects off the unit and time fixed effects together with $X$, and redefining the within-transformed quantities as
\begin{equation}
  \tilde{Y} = M Y, \qquad \tilde{D}_k = M D_k,
  \label{eq:resid_cov}
\end{equation}
which generalizes the double-demeaning of \eqref{eq:flex_dm} to residualize on the covariates as well, the flexible estimator \eqref{eq:flex_var}, the Bayesian model of Section~\ref{sec:bayes}, and the $\ell_0$-PH estimator all apply verbatim with $\tilde{Y}$ and $\tilde{D}$ read as these covariate-residualized quantities. Covariate adjustment thus enters entirely through the within transformation and leaves the partition-selection problem unchanged.

\bigskip

\subsection{Partial Homogeneity}
Let the true parameter vector $\tau = (\tau_1,\ldots,\tau_K)'$ exhibit
\emph{partial homogeneity}: there is a partition
$\mathcal{P} = \{C_1,\ldots,C_m\}$ of $\{1,\ldots,K\}$ into $m \leq K$ groups with
\begin{equation}
  \tau_k = \phi_p \quad \text{for all } k \in C_p, \qquad p = 1,\ldots,m,
  \label{eq:ph_partition}
\end{equation}
where the group effects $\phi_1,\ldots,\phi_m$ are distinct. The groups may have
arbitrary sizes $|C_1|,\ldots,|C_m|$, and we call $\mathcal{P}$ the true
partition. Partial homogeneity is the regime $m < K$, where the true model has
only $m$ distinct parameters and the fully flexible model is over-parameterized
by $K - m$. Nothing below restricts the shape of $\mathcal{P}$; the efficiency
gain accrues to every group $C_p$ with $|C_p| \geq 2$, and a singleton group
contributes none. (The special case of one large group together with singletons
is the sharpest illustration but is not assumed.)

To estimate the correctly specified model under $\mathcal{P}$, define the group
regressors $\Dtd_p = \sum_{k \in C_p} \tilde{D}_{k}$ and collect them into the
$NT \times m$ matrix $\Dtd^* = [\tilde{D}_1, \ldots, \tilde{D}_m]$, where for a
singleton $C_p = \{k_p\}$ we have $\Dtd_p = \tilde{D}_{k_p}$. We assume a
homoskedastic error and proceed with the restricted model
\begin{equation}
  \tilde Y = \Dtd^*\phi + \tilde{\varepsilon},
  \qquad
  \tilde{\varepsilon} \sim \bigl(0,\,\sigma^2 I_r\bigr),
  \label{eq:within}
\end{equation}
whose OLS estimator is
\begin{equation}
  \hat{\phi}_{\text{PH}} = \bigl(\Dtd^{*\prime}\Dtd^*\bigr)^{-1}\Dtd^{*\prime}\tilde Y ,
  \qquad
  \Var\!\bigl(\hat{\phi}\bigr)
  = \sigma^2\bigl(\Dtd^{*\prime}\Dtd^*\bigr)^{-1}.
  \label{eq:oracle_var}
\end{equation}
The $p$-th diagonal element of $\sigma^2(\Dtd^{*\prime}\Dtd^*)^{-1}$ governs the
precision of $\hat\phi_p$, the estimate of the common effect of group $C_p$. For
any $k \in C_p$ the estimate of $\tau_k = \phi_p$ is $\hat\phi_p$, which pools the
identifying variation of every cell in $C_p$, whereas the flexible estimator uses
only the variation in $\tilde{D}_k$. This difference in the effective information
set is the source of the variance gap. We call $\hat{\phi}_{\text{PH}}$ the
partial homogeneity (PH) estimator for the rest of the paper.

We now establish that, group by group, the PH estimator dominates the fully
flexible estimator. Because both are unbiased for $\phi_p$ under the true model,
the Gauss--Markov theorem determines which is more efficient.

\begin{proposition}[PH efficiency dominance]
  \label{prop:oracle_dom}
  Under the true partition $\mathcal{P}$, for every group $C_p$ and every
  $k \in C_p$,
  \begin{equation}
    \Var\!\bigl(\hat\phi_{p}\bigr)
    \;\leq\;
    \Var\!\bigl(\hat\tau_{k,\text{flex}}\bigr),
    \label{eq:var_ineq}
  \end{equation}
  with strict inequality whenever $|C_p| \geq 2$ and the pooled cells contribute
  identifying variation about $\phi_p$ beyond that in $\tilde D_k$ (in particular,
  always under orthogonality with $\|\tilde D_k\|^2 > 0$); under orthogonality a
  singleton group attains equality.
\end{proposition}

\begin{proof}
Under $\mathcal{P}$ we have $\tau_k = \phi_p$ for $k \in C_p$, so
$\hat\tau_{k,\text{flex}}$ is a linear, unbiased estimator of $\phi_p$. By the
Gauss--Markov theorem, $\hat\phi_p$ is the best linear unbiased estimator (BLUE)
of $\phi_p$ among all linear functions of $\tilde Y$, hence
$\Var(\hat\tau_{k,\text{flex}}) \geq \Var(\hat\phi_p)$; this holds for any design.
The inequality is strict whenever $\hat\phi_p \neq \hat\tau_{k,\text{flex}}$ with
positive probability, that is whenever the partition imposes a binding restriction
on the estimation of cell $k$. This holds whenever the other cells in $C_p$ carry
identifying variation about $\phi_p$ beyond that in $\tilde D_k$---generically, and
always under orthogonality with $\|\tilde D_k\|^2 > 0$---but it can fail if a group
member contributes no independent variation (a collinear or unidentified cell, as
in Appendix~\ref{app:identification}), in which case pooling that cell leaves
$\hat\tau_{k,\text{flex}}$ unchanged and the dominance is only weak. Under
orthogonality a singleton group imposes no such restriction and the two
estimators coincide; in a non-orthogonal design, however, pooling the remaining
groups can lower $\Var(\hat\phi_p)$ even for a singleton, so the equality is
special to the orthogonal case.
\end{proof}

\noindent The result holds for every group at once, so the total efficiency gain
is the sum of the within-group gains. Under orthogonality only groups of size at
least two contribute; in non-orthogonal designs even a singleton can gain, since
the restricted estimator borrows strength across the correlated cells
(Section~\ref{sec:cs_application}).

To make the variance reduction explicit, we work under the additional assumption
that the within-transformed cohort-time dummies are mutually orthogonal,
$\tilde{D}_j'\tilde{D}_k = 0$ for $j \neq k$, which holds approximately for
balanced panels in which each unit belongs to exactly one cohort and which we
impose exactly to obtain closed-form expressions. Under orthogonality,
$\Dtd'\Dtd = \mathrm{diag}(n_1,\ldots,n_K)$ with $n_k = \|\tilde{D}_k\|^2$ the
effective sample size for CATT $k$. The flexible variance is then
\begin{equation}
  \Var\!\bigl(\hat\tau_{k,\text{flex}}\bigr)
  = \frac{\sigma^2}{n_k},
  \qquad k = 1,\ldots,K,
  \label{eq:flex_var_orth}
\end{equation}
while $\|\tilde{D}_p\|^2 = \sum_{k\in C_p} n_k$, so the PH variance for group
$C_p$ is
\begin{equation}
  \Var\!\bigl(\hat\phi_{p}\bigr)
  = \frac{\sigma^2}{\displaystyle\sum_{k\in C_p} n_k}.
  \label{eq:oracle_var_orth}
\end{equation}
Taking the ratio for any $k \in C_p$,
\begin{equation}
  \frac{\Var(\hat\phi_{p})}{\Var(\hat\tau_{k,\text{flex}})}
  = \frac{n_k}{\displaystyle\sum_{j\in C_p} n_j}
  \;\leq\; 1,
  \label{eq:var_ratio}
\end{equation}
strictly less than one whenever $|C_p| \geq 2$ and decreasing in the effective
sample sizes of the other cells in the group. In the balanced case $n_k = n$ the
ratio is $1/|C_p|$, so the PH estimator is exactly $|C_p|$ times more efficient
than the flexible estimator for every cell in group $C_p$. The gain grows with
group size, as pooling $|C_p|$ cells concentrates $|C_p|$ times as much effective
sample size onto a single coefficient.

The fully pooled estimator lowers the variance further by imposing a single
coefficient across all $K$ cells, at the cost of bias in every group. Under
Eq.~\ref{eq:within} with orthogonal dummies its probability limit is
\begin{equation}
  \E\bigl[\hat\tau_{\text{pool}}\bigr]
  = \sum_{k=1}^K w_k\,\tau_k
  = \sum_{p=1}^m W_p\,\phi_p,
  \qquad
  w_k = \frac{n_k}{\sum_j n_j}, \quad W_p = \sum_{k\in C_p} w_k,
  \label{eq:pool_plim}
\end{equation}
where $W_p$ is the total weight of group $C_p$. For any cell $k \in C_p$ the bias
is therefore
\begin{equation}
  \E\bigl[\hat\tau_{\text{pool}}\bigr] - \phi_{p}
  = \sum_{q \neq p} W_q\,\bigl(\phi_q - \phi_p\bigr),
  \label{eq:pool_bias_hetero}
\end{equation}
a weight-averaged sum of the gaps between group $p$ and every other group, equal
to $\sum_q W_q\phi_q - \phi_p$. It is non-zero for generic $\phi$ and does not
vanish as $n \to \infty$, but it is zero for a particular group whenever $\phi_p$
equals the weight-averaged group effect $\sum_q W_q\phi_q$, which can occur through
cancellation of groups above and below $\phi_p$. The signals of the other groups bleed into the pooled
estimate, so $\phi_p$ is estimated as a distorted mixture of all $m$ true effects.

The practical challenge is that $\mathcal{P}$ is unknown. A researcher who uses
the fully flexible estimator by default leaves $K - m$ degrees of freedom on the
table and reports unnecessarily wide intervals, while one who collapses to a
single pooled coefficient introduces bias that cannot be detected without further
structure. What is needed is a data-driven procedure that recovers $\mathcal{P}$,
or a close approximation, from the data. The $\ell_0$-penalized estimator and the
Bayesian Dirichlet Process approach of the following sections provide two
principled answers.\footnote{Partial homogeneity also aids identification under limited overlap. A cohort-time cell with no clean comparison, and hence no independent identifying variation, can inherit the common effect of its group whenever that group contains an identified cell, though this is identification by assumption, untestable for the very cell that lacks overlap. Appendix~\ref{app:identification} gives the formal statement.}

\bigskip
\noindent Table~\ref{tab:notation} collects the notation used throughout the paper for reference.

\begin{table}[H]
  \centering
  \caption{Notation used throughout the paper.}
  \label{tab:notation}
  \small
  \begin{tabular}{ll}
    \toprule
    Symbol & Meaning \\
    \midrule
    $N,\ T$ & number of units; number of time periods \\
    $g$ & cohort, i.e.\ period of first treatment ($g=0$: never-treated) \\
    $\mathcal{G}=\{2,\dots,T\}$ & set of treated cohorts \\
    $K$ & number of post-treatment cohort-time (CATT) cells \\
    $\tau_{gt},\ \tau_k$ & CATT of cohort $g$ at time $t$; generic cell effect \\
    $\tau=(\tau_1,\dots,\tau_K)'$ & vector of all $K$ CATTs \\
    $\tilde{Y},\ \tilde{D}$ & within-transformed outcome; full $r\times K$ design \\
    $\tilde{D}_k,\ n_k=\|\tilde{D}_k\|^2$ & column and effective sample size of cell $k$ \\
    $r=NT-N-T+1$ & dimension of the within (residual) space \\
    $\mathcal{P}=\{C_1,\dots,C_m\}$ & partition of the $K$ cells into $m$ groups \\
    $\phi=(\phi_1,\dots,\phi_m)'$ & group effects ($\tau_k=\phi_p$ for $k\in C_p$) \\
    $\tilde{D}^{*}=\tilde{D}_{\mathcal{P}}$ & grouped $r\times m$ design under $\mathcal{P}$ \\
    $\hat\tau_{\text{flex}},\ \hat\tau_{\text{pool}},\ \hat\phi_{\text{PH}}$ & flexible, pooled, and partial-homogeneity estimators \\
    $\alpha$;\ \ $G_0=\mathcal{N}(\mu_0,\sigma_0^2)$ & DP concentration; DP base measure \\
    $\kappa=\sigma_0^2/\sigma^2$ & prior-to-error variance ratio \\
    $c(\mathcal{P})$;\ \ $\lambda$ & number of cross-group CATT pairs; $\ell_0$ penalty \\
    $\hat\Sigma$ & first-stage covariance of $\hat\tau$ (applications) \\
    $\hat\pi_{jk}$ & posterior co-clustering probability of cells $j,k$ \\
    \bottomrule
  \end{tabular}
\end{table}

  \section{A Bayesian Model for Partition Selection}\label{sec:bayes}
  We develop the Dirichlet Process (DP) mixture model for the CATTs, specifying it in Section~\ref{sec:bayes:model}, treating its error variance and the fixed-variance $\ell_0$ maximum-a-posteriori estimator in Section~\ref{sec:bayes:sigmafixed}, and estimating it with a collapsed Gibbs sampler in Section~\ref{sec:bayes:gibbs}.

  \subsection{Model Specification}\label{sec:bayes:model}

Recall the within-transformed model from \eqref{eq:within},
  \begin{equation}
    \tilde{Y} \;=\; \tilde{D}\,\tau \;+\; \tilde{\varepsilon}, \qquad \tilde{\varepsilon} \mid \tilde{D} \sim
  N\!\left(0,\,\sigma^{2} I_r\right).
    \label{eq:bayes:lik}
  \end{equation}
We treat the $\tau_{k}$ as exchangeable draws from an unknown mixing distribution $G$ and place a Dirichlet Process prior on $G$:
  \begin{align}
    G &\;\sim\; \operatorname{DP}(\alpha,\, G_{0}), \label{eq:bayes:dp}\\
    \tau_{k} \mid G &\;\overset{\text{iid}}{\sim}\; G, \qquad k = 1,\dots,K, \label{eq:bayes:tau}\\
    G_{0} &\;=\; \mathcal{N}\!\left(\mu_{0},\,\sigma_{0}^{2}\right). \label{eq:bayes:base}
  \end{align}
  Following \citet{ferguson1973}, draws from a Dirichlet Process are almost surely discrete, so realisations of $G$ assign
  positive probability to ties among the $\tau_{k}$. These ties induce a random clustering of the $K$ CATTs into $m \le K$
  distinct values, which is precisely the partition structure of interest. The base measure $G_{0}$ governs the location of
  cluster effects, and its variance $\sigma_{0}^{2}$ encodes the analyst's prior dispersion of treatment effects across distinct
  groups. In the absence of prior information we recommend a diffuse choice ($\sigma_{0}^{2}$ large relative to $\sigma^{2}$). The concentration parameter $\alpha > 0$ controls the prior expected number of clusters,
  \begin{equation}
    \mathbb{E}[m \mid \alpha, K] \;\approx\; \alpha \log\!\left(1 + K/\alpha\right),
    \label{eq:bayes:mexp}
  \end{equation}
  with $\alpha \to 0$ favoring the fully pooled specification and $\alpha \to \infty$ favoring the fully flexible
  specification. Marginalizing $G$ in \eqref{eq:bayes:dp}--\eqref{eq:bayes:tau} yields a P\'olya-urn predictive distribution for the cluster
  labels $z = (z_{1},\dots,z_{K})$, commonly referred to as the Chinese Restaurant Process (CRP). Conditional on the previous
  $k-1$ assignments,
  \begin{equation}
    \Pr(z_{k} = c \mid z_{1},\dots,z_{k-1}) \;=\;
    \begin{cases}
      \dfrac{n_{c}}{k - 1 + \alpha}, & \text{cluster $c$ has $n_{c}$ members,}\\[6pt]
      \dfrac{\alpha}{k - 1 + \alpha}, & \text{$c$ is a new cluster.}
    \end{cases}
\label{eq:bayes:crp}
\end{equation}
The CRP exhibits the rich-get-richer property, whereby clusters with more members are more likely to attract additional CATTs, but a fresh cluster can always be opened with weight proportional to $\alpha$. Equivalently, the marginal prior probability of any partition $\mathcal{P} = \{C_{1},\dots,C_{m}\}$ of $\{1,\dots,K\}$ takes the closed form
  \begin{equation}
    \Pr(\mathcal{P}) \;=\; \frac{\alpha^{m}\,\prod_{l=1}^{m}\bigl(|C_{l}| - 1\bigr)!}{\alpha^{(K)}}, \qquad \alpha^{(K)} =
  \alpha(\alpha+1)\cdots(\alpha+K-1),
    \label{eq:bayes:partprior}
  \end{equation}
  which depends on $\mathcal{P}$ only through the cluster sizes and the number of clusters. Equation~\eqref{eq:bayes:partprior}
  delivers the partition prior we combine with the data likelihood below.

The complete hierarchical model places a conjugate prior on the error variance as well. Under a partition $\mathcal{P}$, write $\tilde{D}^{*}=\tilde{D}_{\mathcal{P}}$ for the $r\times m$ grouped design of Section~\ref{sec:problem}, whose $m$ columns sum the within-transformed cell columns of each group, so that $\tilde{D}$ denotes the full $K$-column design and $\tilde{D}^{*}$ its partition-collapsed form:
  \begin{align}
    \tilde{y} \mid \mathcal{P}, \phi, \sigma^{2} &\;\sim\; \mathcal{N}\!\left(\tilde{D}^{*}\,\phi,\,\sigma^{2} I_r\right),
  \label{eq:bayes:cond}\\
    \phi \mid \mathcal{P} &\;\sim\; \mathcal{N}\!\left(\mu_{0}\mathbf{1}_{m},\,\sigma_{0}^{2} I_{m}\right), \label{eq:bayes:phi}\\
    \sigma^{2} &\;\sim\; \mathrm{Inv\text{-}Gamma}(a_{0}, b_{0}), \label{eq:bayes:sigma}
  \end{align}
  so that no quantity is fixed by plug-in and the specification is a complete Bayesian model. We take this full model, with $\sigma^{2}$ random, as our primary specification throughout, and it is what the simulations and applications report. The inverse-gamma prior is conditionally conjugate, which keeps the collapsed sampler tractable. In an application the sampler needs only the first-stage cohort-time estimates and their joint covariance, entering through $\tilde{D}^{*\prime}\tilde{D}^{*}$ and $\tilde{D}^{*\prime}\tilde{Y}$, rather than the micro panel; Lemma~\ref{lem:sufficiency} makes this two-stage reading precise and is how the empirical applications of Section~\ref{sec:empirical_application} are run.

  At any value of $\sigma^{2}$, the value the sampler conditions on at a given sweep, the prior on $\phi$ is conjugate to the Gaussian likelihood, so the group effects integrate out analytically.
  Standard calculations yield
  \begin{equation}
    \tilde{y} \mid \mathcal{P} \;\sim\; \mathcal{N}\!\left(\mu_{0}\,\tilde{D}^{*}\,\mathbf{1}_{m},\; \sigma^{2} I_r +
  \sigma_{0}^{2}\,\tilde{D}^{*}\,\tilde{D}^{*\prime}\right).
    \label{eq:bayes:marg}
  \end{equation}
  Direct evaluation of the density implied by \eqref{eq:bayes:marg} would require inverting the $r \times r$ covariance matrix,
   which is computationally prohibitive in panels of empirical scale. The matrix determinant lemma and the Woodbury identity
  reduce the calculation to operations on $m \times m$ matrices. Letting $\kappa = \sigma_{0}^{2}/\sigma^{2}$, writing $\tilde{Y}_{0} = \tilde{Y} - \mu_{0}\,\tilde{D}^{*}\mathbf{1}_{m}$ for the outcome
  centered at the prior mean (its $\tilde{D}^{*}\mathbf{1}_{m}$ is the total treatment indicator, common to every $\mathcal{P}$),
  and dropping additive constants in $\mathcal{P}$, the log marginal likelihood admits the form
  \begin{equation}
    \log p(\tilde{y} \mid \mathcal{P}) \;\propto\; -\tfrac{1}{2}\,\log\bigl|\,I_{m} + \kappa\,\tilde{D}^{*\prime}\tilde{D}^{*}\,\bigr| \;+\;
  \frac{\kappa}{2\sigma^{2}}\, \bigl(\tilde{D}^{*\prime}\tilde{Y}_{0}\bigr)'\!\bigl(I_{m} +
  \kappa\,\tilde{D}^{*\prime}\tilde{D}^{*}\bigr)^{-1}\!\bigl(\tilde{D}^{*\prime}\tilde{Y}_{0}\bigr).
    \label{eq:bayes:logmarg}
  \end{equation}
  The posterior of $\phi$ given $\mathcal{P}$ and $\sigma^{2}$ is $N(\mu^{\ast},\Sigma^{\ast})$ with
  \begin{equation}
    \Sigma^{\ast} \;=\; \left(\tilde{D}^{*\prime}\tilde{D}^{*}/\sigma^{2} + I_{m}/\sigma_{0}^{2}\right)^{\!-1}, \qquad
  \mu^{\ast} \;=\; \Sigma^{\ast}\!\left(\tilde{D}^{*\prime}\tilde{Y}/\sigma^{2} +
  \mu_{0}\mathbf{1}_{m}/\sigma_{0}^{2}\right).
    \label{eq:bayes:phipost}
  \end{equation}
  The collapsed Gibbs sampler (Section~\ref{sec:bayes:gibbs}) draws $\sigma^{2}$ from its inverse-gamma full conditional each sweep and evaluates \eqref{eq:bayes:logmarg}--\eqref{eq:bayes:phipost} at that draw.

  \begin{remark}[The design covariance]
  \label{rem:sigma}
  The grouped posterior \eqref{eq:bayes:phipost} must be built from the exact within-design
  cross-product $\tilde D^{\prime}\tilde D$ of the full flexible model (its grouped version being $\tilde D^{*\prime}\tilde D^{*}$). The cohort-time estimates
  $\hat\tau_{k}$ are jointly $N(\tau, \sigma^2(\tilde D^{\prime}\tilde D)^{-1})$
  and are correlated through the shared fixed effects. Treating them as
  independent understates the posterior variance of linear aggregates such as the
  overall ATT, in our design by a factor of about $3.5$, and produces
  overconfident, undercovering intervals. The model above already uses this exact
  cross-product; we flag the point because the diagonal shortcut is
  tempting and its effect on coverage is large.
  \end{remark}

  \subsection{Fixed and Random $\sigma^{2}$}\label{sec:bayes:sigmafixed}

  Under the full random-$\sigma^{2}$ specification of Section~\ref{sec:bayes:model}, integrating $(\phi,\sigma^{2})$ out of a partition's marginal likelihood gives a scale mixture of Gaussians. Because the base-measure variance $\sigma_{0}^{2}$ is not scaled by $\sigma^{2}$, the marginal covariance is $\sigma^{2} I_r + \sigma_{0}^{2}\tilde D^{*}\tilde D^{*\prime}$ and the mixture is not a standard multivariate-$t$; nonetheless a Schwarz (Laplace) approximation to its log marginal density shows that the MAP partition approximately minimizes a BIC-type criterion on the log residual sum of squares,
  \begin{equation}
    \mathcal{P}^{\mathrm{MAP}} \;\approx\; \arg\min_{\mathcal{P}}\; NT\log\!\bigl(\mathrm{RSS}(\mathcal{P})/NT\bigr) \;+\; m\log(NT) \;-\; 2\log\Pr(\mathcal{P}),
    \label{eq:bic}
  \end{equation}
  so the fit enters on the log scale with $\sigma^{2}$ profiled out, the Occam term $m\log(NT)$ comes from integrating the $m$ group effects and grows at the Schwarz rate, and $\Pr(\mathcal{P})$ is the CRP prior \eqref{eq:bayes:partprior}; the expansion drops terms of order $O(1)$ in $NT$ (Remark~\ref{rem:map}). Holding $\sigma^{2}$ fixed at the plug-in $\hat\sigma^{2} = \mathrm{RSS}_{\text{flex}}/(NT - N - T + 1 - K)$, the residual variance from the fully flexible regression \eqref{eq:twfe_flex}, is a useful special case. With $\sigma^{2}$ constant the model is Gaussian, the group-effect posterior \eqref{eq:bayes:phipost} is Normal, and the sampler draws partitions from the Gaussian marginal likelihood \eqref{eq:bayes:logmarg}. Because the residual degrees of freedom $NT - N - T + 1 - K$ are large, the fixed- and random-$\sigma^{2}$ versions deliver essentially the same coverage and interval length in our experiments (Appendix~\ref{app:sigma}), so the choice is immaterial for the reported results. The homoskedastic assumption itself can be relaxed, and Appendix~\ref{app:hetero} extends the model to a heterogeneous error covariance and augments the sampler with a Gibbs update for group-specific error variances.

  At fixed $\sigma^{2}$ the connection to the $\ell_0$ homogeneity-pursuit estimators of the panel-data and statistics literatures becomes exact rather than asymptotic. It is cleanest under a \emph{pairwise} partition prior that charges each pair of CATTs assigned to different groups,
  \begin{equation}
    \Pr(\mathcal{P}) \;\propto\; \exp\!\bigl(-\lambda_0\, c(\mathcal{P})\bigr),
    \qquad
    c(\mathcal{P}) \;=\; \binom{K}{2} - \sum_{p=1}^{m}\binom{|C_p|}{2},
    \label{eq:pairprior}
  \end{equation}
  where $c(\mathcal{P})$ counts the cross-group CATT pairs. This is a proper prior over partitions, in the same modeling family as Section~\ref{sec:bayes:model}, but it is not the CRP prior \eqref{eq:bayes:partprior}: the CRP weights a partition through $\alpha^{m}\prod_{l}(|C_{l}|-1)!$, which favors a few large clusters, whereas \eqref{eq:pairprior} imposes a flat cost per separated pair. We use \eqref{eq:pairprior} for the point estimator and retain the CRP for the sampler of Section~\ref{sec:bayes:gibbs}.

  \begin{proposition}[$\ell_0$-PH as a fixed-variance MAP]
  \label{prop:l0map}
  Fix $\sigma^{2}$, place an improper flat prior $p(\phi\mid\mathcal{P})\propto 1$ on the group effects, and adopt the pairwise prior \eqref{eq:pairprior}. Then the joint maximum a posteriori estimator of $(\mathcal{P},\phi)$ has group-effect component equal to the restricted least-squares estimate $\hat\phi=(\Dtd^{*\prime}\Dtd^{*})^{-1}\Dtd^{*\prime}\tilde Y$ under $\mathcal{P}$, and partition component
  \begin{equation}
    \hat{\mathcal{P}} \;=\; \arg\min_{\mathcal{P}}\; Q(\mathcal{P}), \qquad
    Q(\mathcal{P}) \;=\; \mathrm{RSS}(\mathcal{P}) \;+\; \lambda\, c(\mathcal{P}), \qquad \lambda = 2\sigma^{2}\lambda_0 .
    \label{eq:l0map}
  \end{equation}
  \end{proposition}

  \begin{proof}
  For fixed $\mathcal{P}$ the flat prior leaves the Gaussian likelihood, whose profile over $\phi$ is maximized at the restricted least-squares estimate and equals $\exp\!\bigl(-\mathrm{RSS}(\mathcal{P})/2\sigma^{2}\bigr)$ up to a factor constant in $\mathcal{P}$. Multiplying by \eqref{eq:pairprior} and taking $-2\sigma^{2}\log$ gives $\mathrm{RSS}(\mathcal{P})+2\sigma^{2}\lambda_0\,c(\mathcal{P})$ up to an additive constant, which is $Q(\mathcal{P})$ with $\lambda=2\sigma^{2}\lambda_0$.
  \end{proof}

  Written out in the outcome equation, $Q(\mathcal{P})$ trades fit against an $\ell_0$ penalty on the cross-group CATT pairs, pairs of CATTs placed in different groups and hence estimated with different coefficients,
  \begin{equation}
  Q(\mathcal{P}) \;=\; \underbrace{\sum_i \sum_t \Bigl(Y_{igt} - \alpha_i - \gamma_t - \sum_{g \in \mathcal{G}}\sum_{t \geq g}
  \tau_{gt}\, D_{igt}\Bigr)^2}_{\mathrm{RSS}(\mathcal{P})} \;+\; \underbrace{\lambda\sum_{g \neq g'} \mathbf{1}[\tau_{gt} \neq
  \tau_{g't}]}_{\ell_0 \text{ penalty}},
    \label{eq:l0obj}
  \end{equation}
  so the $\ell_0$-PH estimator is the MAP of the model \eqref{eq:pairprior} rather than a separate procedure. Here $\lambda > 0$ is the penalty parameter and $\mathbf{1}$ the indicator function. The penalty discourages splitting
  groups, which increases $c(\mathcal{P})$, unless the resulting reduction in RSS is large enough to compensate. When $\lambda =
   0$, $Q(\mathcal{P}) = \mathrm{RSS}(\mathcal{P})$, which is minimized by the fully flexible partition with $K$ groups. As
  $\lambda \to \infty$, the penalty dominates and $Q(\mathcal{P})$ is minimized by the fully pooled partition with $1$ group. For
   intermediate $\lambda$, the optimal partition reflects the data's evidence on which CATTs are equal. Merge groups $A$ and $B$
  if and only if the increase in RSS from constraining them equal is smaller than $\lambda \cdot |A| \cdot |B|$, the factor $|A|
  \cdot |B|$ being the number of cross-group pairs a merge eliminates.

  \begin{remark}[Which MAP is which]
  \label{rem:map}
  Four related objects should be kept distinct. (i)~The \emph{joint} MAP of Proposition~\ref{prop:l0map}, under a flat base measure and the pairwise prior \eqref{eq:pairprior}, is the exact $\ell_0$ objective \eqref{eq:l0obj}. (ii)~The \emph{marginal} MAP that integrates the Gaussian group effects out of \eqref{eq:bayes:logmarg} retains a log-determinant and a shrinkage term and does not reduce to a residual sum of squares; in the diffuse-base limit $\sigma_{0}^{2}\to\infty$ its complexity charge is $\tfrac{m}{2}\log(\sigma_{0}^{2}/\sigma^{2})$, a penalty on the number of groups rather than on cross-group pairs. (iii)~Integrating the variance out as well yields a scale mixture of Gaussians (not a standard multivariate-$t$, since the base-measure variance $\sigma_{0}^{2}$ is not scaled by $\sigma^{2}$), whose Schwarz (Laplace) approximation is the BIC criterion \eqref{eq:bic}, exact only up to $O(1)$ terms. (iv)~The agglomerative search of Appendix~\ref{app:algorithm} minimizes \eqref{eq:l0obj} up to the orthogonal approximation of its merge cost. The point estimator we report is~(i), computed by the two-step procedure below; the sampler of Section~\ref{sec:bayes:gibbs} targets the full posterior under the CRP prior and relies on none of these approximations.
  \end{remark}

  \begin{proposition}[Threshold interpretation]
  \label{prop:threshold}
  For two singleton groups $\{j\}$ and $\{k\}$, the merge criterion $\Delta\mathrm{Obj}(j,k) < 0$ is equivalent to
  \begin{equation}
    \bigl(\hat{\tau}_j - \hat{\tau}_k\bigr)^2 \cdot \frac{n_j\, n_k}{n_j + n_k} \;<\; \lambda,
    \label{eq:threshold}
  \end{equation}
  where $n_j = \|\tilde{D}_j\|^2$ is the effective sample size for CATT $j$. Two CATTs are pooled if and only if their squared
  difference, scaled by their harmonic-mean sample size, falls below the threshold $\lambda$.
  \end{proposition}

  \noindent The merge cost in \eqref{eq:threshold} is the orthogonal (diagonal-approximation) RSS increment of
  Appendix~\ref{app:algorithm}: it is exact when the within-transformed dummies are orthogonal, whereas in a
  non-orthogonal design the exact increment depends on the full cross-product $\Dtd'\Dtd$, and the Step-2 OLS
  re-estimation restores the exact fit in either case. Read as a rate calculation, the result clarifies why a
  wide range of $\lambda$ separates correct from incorrect merges: the cost of a correct merge (two CATTs with
  the same true effect) is $O_p(\sigma^2)$, the harmonic-weighted squared difference of two estimates separated
  only by noise, distributed as $\sigma^2$ times a $\chi^2_1$, whereas the cost of a wrong merge (different true
  effects) is $O(N \cdot \Delta\tau^2)$, where $\Delta\tau$ is the effect gap. Because the correct-merge cost is
  $O_p(\sigma^2)$ rather than vanishing, a fixed $\lambda = c\,\sigma^2$ leaves a constant probability of
  over-splitting (the upper $\chi^2_1$ tail); recovering $\mathcal{P}$ with probability approaching one for fixed
  separated effects therefore requires a threshold that grows, with $\sigma^2 = o(\lambda_N)$ to suppress spurious
  splits and $\lambda_N = o(N\,\Delta\tau^2)$ to prevent wrong merges. The BIC choice, whose effective per-merge
  threshold is $\lambda \asymp \sigma^2\log(NT)$, sits in this window, and the window widens with the separation;
  a formal consistency statement is left to Section~\ref{sec:conc}.

  The penalty in \eqref{eq:l0obj} counts the number of distinct pairs across groups, which is the $\ell_0$ norm applied to the
  vector of all $\binom{K}{2}$ pairwise differences $\{\tau_j - \tau_k\}_{j < k}$. The fused LASSO \citep{Tibshirani2005} instead
   uses the $\ell_1$ norm of pairwise differences (for ordered indices), producing a convex but inexact version of equality. The
  advantage of the $\ell_0$ formulation is that it produces exact equality constraints rather than approximate shrinkage.

  Computing the estimator takes two steps, an agglomerative search that returns $\hat{\mathcal{P}}$ in $O(K^3)$ time and an ordinary least squares re-estimation of the grouped model under $\hat{\mathcal{P}}$.\footnote{Appendix~\ref{app:algorithm} gives the agglomerative merge criterion and the re-estimation, and explains why the second step is needed. The agglomerative search minimizes an orthogonal approximation to the objective, and its pooled coefficients equal the restricted OLS only when the within-transformed cohort-time dummies are orthogonal, so the OLS re-estimation recovers the BLUE in general panels.}

  \subsection{Estimation and Inference}\label{sec:bayes:gibbs}

  We sample from the posterior $\Pr(\mathcal{P} \mid \tilde{y}) \propto p(\tilde{y} \mid \mathcal{P})\Pr(\mathcal{P})$ using the
  collapsed Gibbs sampler of \citet{neal2000} (Algorithm~3), which integrates $\phi$ out analytically at each cluster-assignment
  step and, at the end of every sweep, reinstates it and updates the error variance $\sigma^{2}$ by conjugate draws. At each iteration, we cycle over all $K$ CATTs. For CATT
   $k$:
  \begin{enumerate}
  \item Remove $k$ from its current cluster; delete the cluster if empty.
  \item For each existing cluster $c$ and for a new cluster, compute the conditional probability
  \begin{equation}
    \Pr(z_{k} = c \mid z_{-k},\tilde{y}) \;\propto\;
    \begin{cases}
      n_{-k,c}\,\times\, p(\tilde{y} \mid z_{k} = c,\, z_{-k}), & \text{existing,}\\[4pt]
      \alpha\,\times\, p(\tilde{y} \mid z_{k} = c_{\text{new}},\, z_{-k}), & \text{new,}
    \end{cases},
    \label{eq:bayes:condprob}
  \end{equation}
  where the marginal likelihoods are computed via \eqref{eq:bayes:logmarg}.
  \item Draw $z_{k}$ from the normalized probabilities \eqref{eq:bayes:condprob}.
  \end{enumerate}
  After updating all cluster assignments, draw $\phi$ from the conjugate posterior
  $N(\mu^{\ast},\Sigma^{\ast})$ and then, in the primary random-variance model, draw $\sigma^{2}$ from its inverse-gamma full
  conditional given the current partition and group effects (the conjugate update of the prior \eqref{eq:bayes:sigma}); the
  cluster-assignment marginals \eqref{eq:bayes:logmarg} of the next sweep are evaluated at this draw. The sequence of draws
  $\{z^{(s)},\phi^{(s)},\sigma^{2(s)}\}_{s=1}^{S}$ approximates the posterior. Fixing $\sigma^{2}$ at the plug-in $\hat\sigma^{2}$
  omits this final draw and yields the fixed-variance sampler of Section~\ref{sec:bayes:sigmafixed}.

The posterior mean of $\tau_{k}$ is
  \begin{equation}
    \hat{\tau}_{k}^{\,\text{Bayes}} \;=\; (S-S_B)^{-1}\sum_{s = S_B+1}^{S}\,\phi^{(s)}_{z_{k}^{(s)}},
    \label{eq:bayes:tauhat}
  \end{equation}
  where $S_B$ is the burn-in sample. The $95\%$ credible interval is obtained from the empirical quantiles of the draws. The co-clustering probability
  \begin{equation}
    \hat{\pi}_{jk} \;=\; (S-S_B)^{-1}\sum_{s = S_B+1}^{S}\, \mathbf{1}\!\bigl\{z_{j}^{(s)} = z_{k}^{(s)}\bigr\}
    \label{eq:bayes:cocluster}
  \end{equation}
  gives the posterior probability that CATTs $j$ and $k$ belong to the same group. The $K \times K$ matrix $\hat{\Pi} = [\hat{\pi}_{jk}]$ is the posterior similarity matrix, summarizing the full uncertainty over the grouping structure; when a single representative partition is desired, it can be obtained from the posterior together with a credible region using the loss-based summaries of \citet{lau2007} and \citet{wade2018}. Aggregated estimands (e.g.\ the overall ATT or event-study coefficients) are computed by applying the linear aggregator of interest to each posterior draw and summarizing across the chain, which delivers uncertainty quantification that incorporates the uncertainty in the partition itself.

  The MAP partition of Section~\ref{sec:bayes:sigmafixed} and the posterior mean differ in what they report. The MAP is a single grouping of the cohort-time effects into a few distinct levels, the minimizer of the penalized residual sum of squares. The posterior mean, by contrast, averages over partitions and returns values that match no single grouping, so extracting a representative partition from it requires a separate loss-based summary of the co-clustering matrix.

  We do not claim that the mode is a less biased point estimate than the posterior mean. Holding the regularization fixed, the two nearly coincide, so choosing between the MAP and the posterior mean is about whether a single partition or the full partition-averaged summary is wanted, not about small-sample bias, and inference should come from the posterior in either case.

  \begin{remark}[Inference after partition selection]
  \label{rem:postselection}
  The variance formula for $\hat\phi$ in \eqref{eq:phgeneral} treats the
  partition $\hat{\mathcal{P}}$ as fixed. In finite samples $\hat{\mathcal{P}}$
  is itself estimated, so confidence intervals built from
  $\sigma^2[(\tilde D^{*\prime}\tilde D^*)^{-1}]_{pp}$ are valid only
  conditional on the selected partition and can under-cover
  unconditionally, as with inference after any model-selection step \citep{leeb2005}.
  When the recovery interval of Proposition~\ref{prop:threshold} is wide, so
  that the true partition is selected with probability approaching one, this
  gap is small, and the simulations of Section~\ref{sec:sim} confirm that the
  plug-in intervals are then only mildly anti-conservative. The Bayesian
  estimator of Section~\ref{sec:bayes} avoids the issue entirely by averaging
  over the partition rather than conditioning on a single one, and attains
  near-nominal coverage in our experiments; it is our preferred route to honest
  uncertainty quantification.\footnote{An obvious alternative, sample
  splitting, selecting the partition on one subsample and estimating the
  coefficients on another, does not help, because halving the sample degrades
  recovery and a wrongly merged partition induces a specification bias that no
  standard-error correction removes. A selective-inference correction that
  conditions on the selection event \citep{fithian2014, rinaldo2019} is a further
  route we do not pursue.}
  \end{remark}
\section{Simulation Study}
\label{sec:sim}

This section reports a Monte Carlo study with three aims: (i) to document, on the
variance scale that matters for causal estimands, the efficiency of the
Dirichlet-Process (DP) estimator relative to the fully flexible and fully
pooled benchmarks under partial homogeneity; (ii) to show how those gains depend
on how well separated the distinct effects are; and (iii) to evaluate honest
inference by measuring confidence- and credible-interval coverage. We also carry
the $\ell_0$ MAP estimator through the same experiments to see how the model's
mode compares with its posterior, and it behaves much like the DP throughout.

\subsection{Design}
\label{sec:sim:design}

We simulate a balanced panel with $N = 2{,}000$ units and $T = 10$ periods.
Units are split equally across four cohorts, a never-treated group ($g=0$) and
three treated cohorts entering at $t = 3,5,7$. This yields $K = 18$
post-treatment cohort-time cells. Outcomes follow
\begin{equation}
  Y_{it} = \alpha_i + \lambda_t
    + \tau^{\ast}_{g(i),t}\,\mathbf{1}\{t \ge g(i)\} + \varepsilon_{it},
  \qquad
  \alpha_i \sim \mathcal N(0,1),\ \lambda_t = 0.1(t-1),\ \varepsilon_{it}\sim\mathcal N(0,1).
  \label{eq:sim:dgp}
\end{equation}
The unit and time fixed effects are removed by the within transformation, so the
estimators' sampling behavior is driven entirely by $\varepsilon$.

The partial-homogeneity truth is parameterized by the number of distinct effects
and their separation. For a given number of distinct effects $m^{\ast}$, we partition the $K=18$ cells
into $m^{\ast}$ near-equal groups and assign each group a mean on an evenly
spaced grid with adjacent gap $\Delta$. We sweep
$m^{\ast}\in\{1,3,6,9,18\}$, from fully homogeneous ($m^{\ast}=1$, where pooling
is correct) to fully heterogeneous ($m^{\ast}=18$, where the flexible estimator
is correct). The gap is calibrated to a unit-free separation
$\delta \equiv \Delta / \overline{\mathrm{sd}}(\hat\tau_{\text{flex}})$, the
distance between adjacent group means in standard errors of the flexible
cohort-time estimates; $\delta$ governs how easily the groups can be told apart.
We report $\delta \in \{3,6,12\}$, with $\delta=6$ as the headline.

We compare five estimators, pooled TWFE (the biased-but-precise benchmark),
flexible TWFE (unbiased-but-imprecise, in this balanced design it coincides with
the \citet{gardner2022two} two-stage estimator), the infeasible oracle
partial-homogeneity estimator that is handed the true partition (the efficiency
floor), the $\ell_0$-PH estimator, and the Bayes-PH estimator. The $\ell_0$ penalty is
chosen by BIC \citep{schwarz1978} along the agglomeration path,
$\mathrm{BIC}(m) = NT\log(\mathrm{RSS}(\hat{\mathcal P}_m)/NT) + m\log(NT)$,
the flat-partition-prior case of \eqref{eq:bic}, requiring no tuning grid. The DP sampler uses a diffuse base measure
($\mu_0=0$, $\sigma_0^2=100$) and concentration $\alpha=7$ (so the prior expected
number of groups is about $K/2$), run for $150$ Gibbs sweeps with a $50$-sweep
burn-in, drawing $\sigma^{2}$ from its inverse-gamma full conditional each sweep
(the full Bayesian model), and using the exact within-design covariance of the
cohort-time estimates throughout, in both the cluster-assignment moves and the
recorded draws (Remark~\ref{rem:sigma}). The overall ATT is the
equal-weight average of the $K$ CATTs; aggregator choice is orthogonal to the
specification question and does not affect the conclusions. All results use
$500$ Monte Carlo replications.

\subsection{Variance and recovery under partial homogeneity}
\label{sec:sim:var}

Table~\ref{tab:sim:var} reports, for the headline separation $\delta=6$, the
sampling variance of the cohort-time estimates relative to flexible TWFE (lower
is more precise), the average absolute CATT bias (a unbiasedness check), and the
adjusted Rand index (ARI) measuring how well the true partition is recovered.

\begin{table}[H]
  \centering
  \caption{Sampling variance (ratio to flexible), bias, and partition recovery
  under partial homogeneity, $\delta=6$. ``Var ratio'' is the average per-cell
  Monte Carlo variance divided by that of flexible TWFE; ``$|$Bias$|$'' is the
  average absolute CATT bias; ARI is the adjusted Rand index against the true
  partition. Oracle PH is infeasible (true partition known). $500$ replications.}
  \label{tab:sim:var}
  \small
  \begin{tabular}{llccc}
    \toprule
    Scenario & Method & Var ratio & $|$Bias$|$ & ARI \\
    \midrule
    \multirow{5}{*}{$m^{\ast}=1$ (homogeneous)}
      & Pooled     & 0.14 & 0.001 & 1.00 \\
      & Flexible   & 1.00 & 0.004 & 0.00 \\
      & Oracle PH  & 0.14 & 0.001 & 1.00 \\
      & $\ell_0$-PH   & 0.50 & 0.002 & 0.40 \\
      & Bayes-PH   & 0.16 & 0.001 & 0.77 \\
    \midrule
    \multirow{5}{*}{$m^{\ast}=3$ (partial)}
      & Pooled     & 0.13 & 0.236 & 0.00 \\
      & Flexible   & 1.00 & 0.002 & 0.00 \\
      & Oracle PH  & 0.26 & 0.001 & 1.00 \\
      & $\ell_0$-PH   & 0.56 & 0.003 & 0.94 \\
      & Bayes-PH   & 0.48 & 0.005 & 0.91 \\
    \midrule
    \multirow{5}{*}{$m^{\ast}=6$ (partial)}
      & Pooled     & 0.14 & 0.490 & 0.00 \\
      & Flexible   & 1.00 & 0.002 & 0.00 \\
      & Oracle PH  & 0.42 & 0.002 & 1.00 \\
      & $\ell_0$-PH   & 0.74 & 0.004 & 0.95 \\
      & Bayes-PH   & 0.70 & 0.005 & 0.89 \\
    \midrule
    \multirow{5}{*}{$m^{\ast}=9$ (partial)}
      & Pooled     & 0.14 & 0.741 & 0.00 \\
      & Flexible   & 1.00 & 0.003 & 0.00 \\
      & Oracle PH  & 0.64 & 0.002 & 1.00 \\
      & $\ell_0$-PH   & 1.26 & 0.005 & 0.91 \\
      & Bayes-PH   & 1.63 & 0.009 & 0.77 \\
    \midrule
    \multirow{5}{*}{$m^{\ast}=18$ (heterogeneous)}
      & Pooled     & 0.15 & 1.483 & 0.00 \\
      & Flexible   & 1.00 & 0.002 & 1.00 \\
      & Oracle PH  & 1.00 & 0.002 & 1.00 \\
      & $\ell_0$-PH   & 3.38 & 0.008 & 0.01 \\
      & Bayes-PH   & 2.36 & 0.009 & 0.05 \\
    \bottomrule
  \end{tabular}
\end{table}

Three patterns stand out. First, pooled TWFE is always the most precise (variance
ratio $\approx 0.14$) but is severely biased under any heterogeneity, with an
average CATT bias rising from $0.24$ at $m^{\ast}=3$ to $1.48$ at $m^{\ast}=18$;
it is unusable for cohort-time effects. Second, at the partial-homogeneity
scenarios that motivate the paper ($m^{\ast}=3,6$) the $\ell_0$-PH and Bayes-PH estimators
cut the sampling variance of the CATTs by $26$--$52\%$ relative to flexible
TWFE while remaining essentially unbiased (average $|$bias$|\le 0.005$), and
they recover the partition well (ARI $\approx 0.9$). Third, the methods correctly
degrade at the extremes, in that at $m^{\ast}=1$ they shrink toward pooling (though not
all the way), and at $m^{\ast}=18$, where every cell is genuinely distinct, any
pooling hurts and the flexible estimator is preferred. At this moderate separation
the feasible estimators do not fully revert to flexible at $m^{\ast}=18$ but carry a
selection overhead (variance ratios $3.4$ and $2.4$), the price of choosing the
partition when there is no poolable structure; the overhead disappears once the
effects are well separated, where at $\delta=12$, $m^{\ast}=18$ both recover all $K$
cells and match flexible (variance ratios $1.00$ and $1.00$). Figure~\ref{fig:sim:var}
plots the variance ordering across $m^{\ast}$.

\begin{figure}[H]
  \centering
  \includegraphics[width=0.7\linewidth]{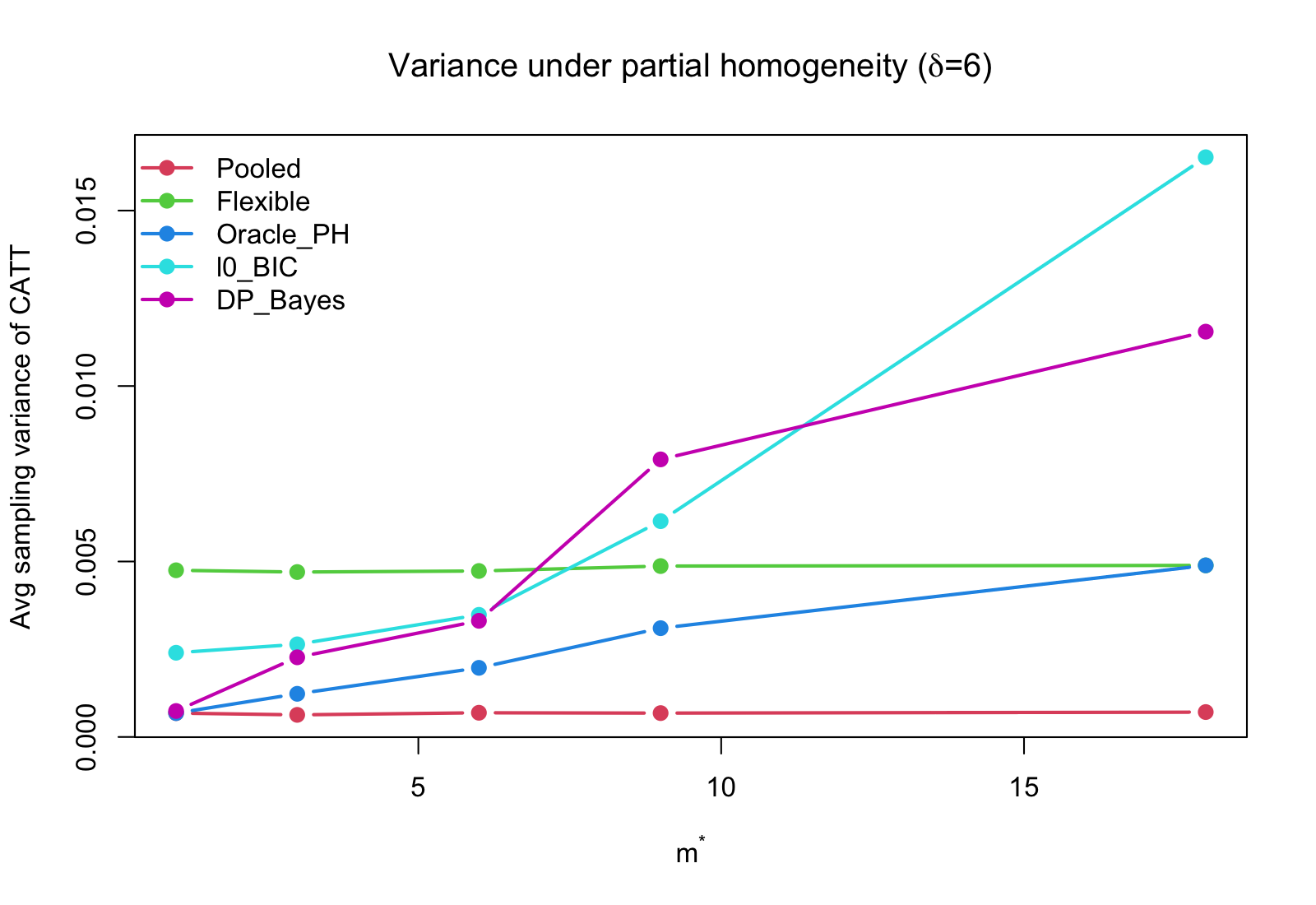}
  \caption{Average sampling variance of the cohort-time estimates against the true
  number of distinct effects $m^{\ast}$ ($\delta=6$). The $\ell_0$ and DP
  estimators track the infeasible oracle in the partial-homogeneity region
  ($m^{\ast}=3,6$); when every cell is genuinely distinct ($m^{\ast}\to K$) at this
  moderate separation they carry a selection overhead above flexible TWFE, which
  vanishes once the effects are well separated ($\delta=12$; see
  Table~\ref{tab:sim:sep} and the text).}
  \label{fig:sim:var}
\end{figure}

The realized gain depends on how well separated the distinct effects are, exactly
as the recovery interval of Proposition~\ref{prop:threshold} suggests.
Table~\ref{tab:sim:sep} fixes $m^{\ast}=6$ and varies $\delta$.

\begin{table}[H]
  \centering
  \caption{Effect of separation $\delta$ on precision and recovery,
  $m^{\ast}=6$. Variance ratio to flexible and ARI. $500$ replications.}
  \label{tab:sim:sep}
  \small
  \begin{tabular}{lcccccc}
    \toprule
    & \multicolumn{2}{c}{$\delta=3$} & \multicolumn{2}{c}{$\delta=6$}
    & \multicolumn{2}{c}{$\delta=12$} \\
    \cmidrule(lr){2-3}\cmidrule(lr){4-5}\cmidrule(lr){6-7}
    Method & Var ratio & ARI & Var ratio & ARI & Var ratio & ARI \\
    \midrule
    Oracle PH & 0.40 & 1.00 & 0.42 & 1.00 & 0.38 & 1.00 \\
    $\ell_0$-PH  & 1.46 & 0.53 & 0.74 & 0.95 & 0.41 & 1.00 \\
    Bayes-PH  & 1.34 & 0.40 & 0.70 & 0.89 & 0.42 & 0.94 \\
    \bottomrule
  \end{tabular}
\end{table}

When the effects are close ($\delta=3$) the partition cannot be recovered
reliably (ARI $\approx 0.5$) and the selection noise erases, indeed reverses, the
variance advantage, and both feasible estimators are then less precise than
flexible TWFE. As separation grows the gap to the oracle closes, and at
$\delta=12$ the feasible $\ell_0$-PH estimator essentially attains the oracle
variance reduction (0.41 versus 0.38) with near-perfect recovery.
Figure~\ref{fig:sim:sep} shows the monotone relationship. The practical message
is that the methods deliver their promised gains when heterogeneity is
``clumpy'', a modest number of well-separated effect levels, and should not
be expected to when the effects form a near-continuum. In the hard regimes
neither feasible estimator beats flexible TWFE on variance, and both can fall
behind it, at $\delta=3$ (variance ratios $1.46$ for $\ell_0$-PH and $1.34$ for
Bayes-PH) and at $\delta=6$ with the fine $m^{\ast}=9$ partition; the efficiency
gains are specific to the clumpy regime. Where the two estimators differ sharply
is in inference rather than variance: Section~\ref{sec:sim:cov} shows the Bayes-PH
intervals stay far better calibrated than the $\ell_0$-PH plug-in when the
partition is uncertain.

\begin{figure}[H]
  \centering
  \includegraphics[width=\linewidth]{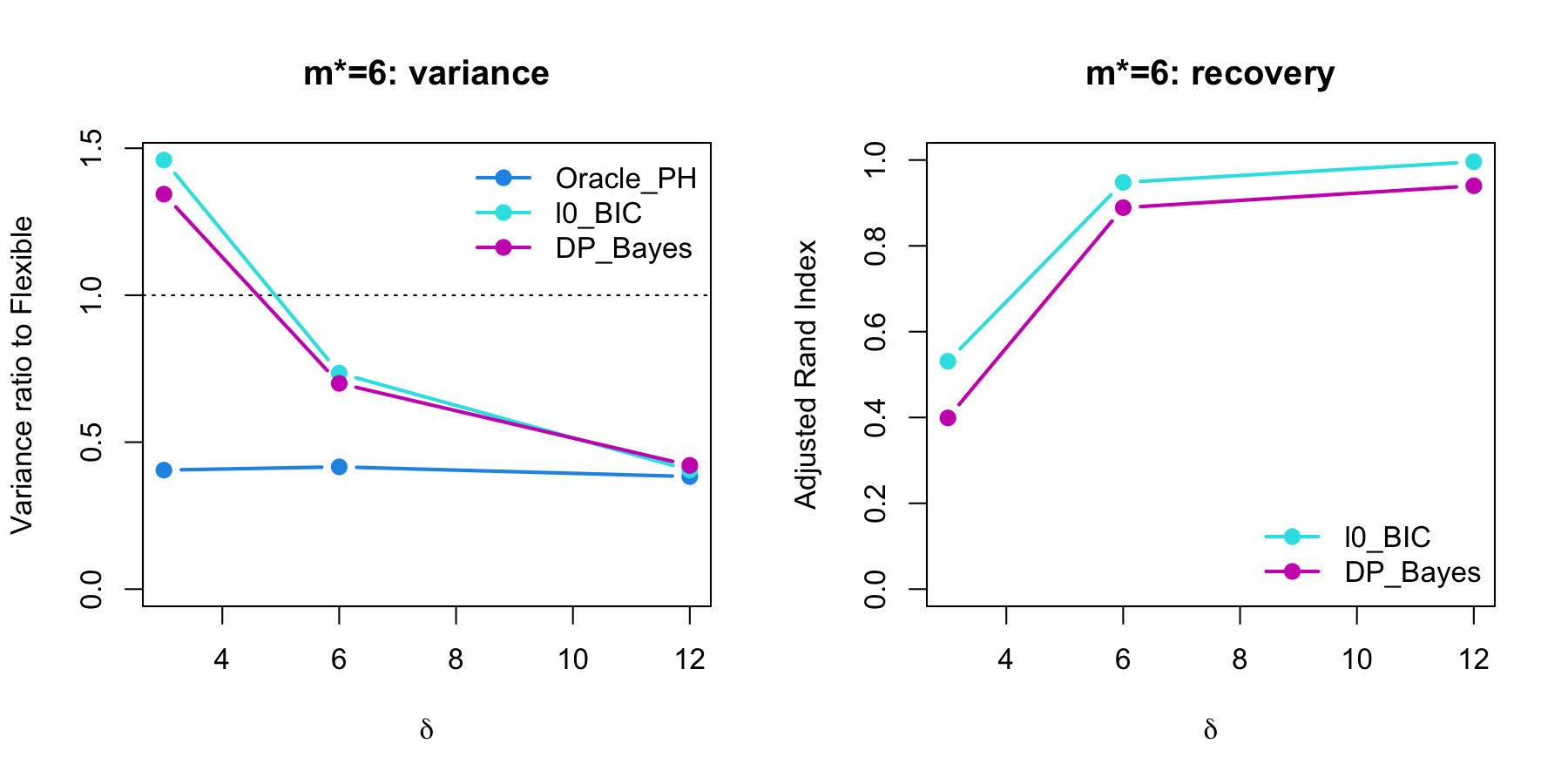}
  \caption{Precision (left) and partition recovery (right) as functions of the
  separation $\delta$, at $m^{\ast}=6$. The variance advantage over flexible TWFE
  (values below the dashed line) is realized only once the effects are separated
  enough to be recovered.}
  \label{fig:sim:sep}
\end{figure}

\subsection{Honest inference: coverage}
\label{sec:sim:cov}

Table~\ref{tab:sim:cov} evaluates the calibration of $95\%$ intervals at
$m^{\ast}=6$, $\delta=6$, reporting coverage and average length for the
cohort-time effects and for the overall ATT. For the $\ell_0$-PH estimator we
report the plug-in interval that treats the selected partition as fixed. The
Bayes-PH interval is the $2.5$--$97.5\%$ posterior credible interval.

\begin{table}[H]
  \centering
  \caption{Coverage and average length of nominal $95\%$ intervals at
  $m^{\ast}=6$, $\delta=6$; targets are the true CATTs and the true ATT. Oracle PH
  is infeasible. $500$ replications.}
  \label{tab:sim:cov}
  \small
  \begin{tabular}{lcccc}
    \toprule
    & \multicolumn{2}{c}{CATT} & \multicolumn{2}{c}{ATT} \\
    \cmidrule(lr){2-3}\cmidrule(lr){4-5}
    Method & Coverage & Length & Coverage & Length \\
    \midrule
    Pooled                 & 0.04 & 0.10 & 0.05 & 0.10 \\
    Flexible               & 0.95 & 0.27 & 0.94 & 0.12 \\
    Oracle PH              & 0.95 & 0.17 & 0.95 & 0.11 \\
    $\ell_0$-PH             & 0.90 & 0.17 & 0.92 & 0.11 \\
    Bayes-PH               & 0.93 & 0.20 & 0.92 & 0.11 \\
    \bottomrule
  \end{tabular}
\end{table}

The pooled intervals collapse (coverage $0.04$--$0.05$): they are short but
centered on a biased estimand. The flexible and oracle intervals are correctly
calibrated, the oracle being much shorter, and they bracket the feasible
estimators. The Bayes-PH credible intervals attain near-nominal coverage
($0.93$ for the CATTs, $0.92$ for the ATT) at a length between oracle and
flexible, honest and informative uncertainty. The $\ell_0$-PH interval is
mildly anti-conservative here ($0.90$ CATT, $0.92$ ATT): at this separation
the correct partition is recovered often enough that the conditional-on-selection
distortion is modest.\footnote{A natural alternative removes this distortion by
sample splitting, selecting the partition on one half of the units and computing
coefficients and standard errors on the other. It does not help, its coverage is
in fact lower ($0.75$ CATT, $0.90$ ATT). Splitting halves the sample, which
degrades recovery, and a wrongly merged partition then injects a specification
bias that a correctly sized standard error cannot undo.} Honest uncertainty under
partition selection thus comes not from adjusting the variance but from averaging
over the partition, which is what the DP posterior does.
It is also worth noting what does drive the result. The calibration hinges on
using the exact cross-cell covariance of the cohort-time estimates
(Remark~\ref{rem:sigma}), whereas the treatment of $\sigma^2$ is immaterial. The
full Bayesian sampler draws $\sigma^2$ from its inverse-gamma full conditional, and
holding it fixed at a plug-in value gives essentially the same coverage and
interval length (Appendix~\ref{app:sigma}).

The headline design is favorable, so Table~\ref{tab:sim:covgrid} extends the
coverage check across the grid, from the homogeneous case ($m^{\ast}=1$) through
partial homogeneity ($m^{\ast}=6$) to full heterogeneity ($m^{\ast}=18$), and
across separations $\delta\in\{3,6,12\}$. Two patterns matter. First, the overall
ATT is well calibrated everywhere for both feasible estimators, with Bayes-PH
coverage between $0.91$ and $0.96$ across every cell, because aggregation averages
the cell-level errors. Second, the cohort-time (CATT) intervals are where the
regime bites, and the two estimators part ways. When the effects are separable
($\delta=12$, or $\delta=6$ at $m^{\ast}=6$) both attain near-nominal CATT
coverage (the $\delta=6$, $m^{\ast}=6$ cell reproduces the headline
Table~\ref{tab:sim:cov} within Monte Carlo error). In the hard regimes, at low separation ($\delta=3$) or full
heterogeneity ($m^{\ast}=18$), Bayes-PH degrades gracefully, to $0.79$--$0.81$,
while the $\ell_0$-PH plug-in interval collapses to $0.57$--$0.62$: committing to
one wrong partition and treating it as known is far more damaging than averaging
over partitions, which is again the DP posterior's advantage.

\begin{table}[H]
  \centering
  \caption{Coverage of nominal $95\%$ intervals across the $(m^{\ast},\delta)$
  grid, for the cohort-time effects (CATT) and the overall ATT. $m^{\ast}=1$ is
  homogeneous (separation irrelevant). Flexible is the unbiased benchmark;
  $\ell_0$-PH uses the plug-in interval that treats the selected partition as
  fixed; Bayes-PH is the posterior credible interval. $500$ replications.}
  \label{tab:sim:covgrid}
  \small
  \begin{tabular}{llccccccc}
    \toprule
    & & \multicolumn{3}{c}{CATT coverage} & & \multicolumn{2}{c}{ATT coverage} \\
    \cmidrule(lr){3-5}\cmidrule(lr){7-9}
    Scenario & $\delta$ & Flexible & $\ell_0$-PH & Bayes-PH & & Flexible & $\ell_0$-PH & Bayes-PH \\
    \midrule
    $m^{\ast}=1$ (homog.)      & --- & 0.95 & 0.76 & 0.94 & & 0.94 & 0.92 & 0.93 \\
    \midrule
    \multirow{3}{*}{$m^{\ast}=6$ (partial)}
      & 3  & 0.95 & 0.62 & 0.81 & & 0.95 & 0.89 & 0.91 \\
      & 6  & 0.95 & 0.90 & 0.94 & & 0.96 & 0.94 & 0.94 \\
      & 12 & 0.95 & 0.95 & 0.94 & & 0.94 & 0.95 & 0.92 \\
    \midrule
    \multirow{3}{*}{$m^{\ast}=18$ (heterog.)}
      & 3  & 0.95 & 0.58 & 0.79 & & 0.95 & 0.91 & 0.91 \\
      & 6  & 0.95 & 0.57 & 0.80 & & 0.94 & 0.90 & 0.92 \\
      & 12 & 0.95 & 0.95 & 0.94 & & 0.96 & 0.96 & 0.94 \\
    \bottomrule
  \end{tabular}
\end{table}

\subsection{Solution paths}
\label{sec:sim:path}

Figure~\ref{fig:sim:path} traces the bias-variance trade-off as a function of
the regularization strength, for $\ell_0$-PH along the agglomeration path indexed by the
number of groups $m$ (small $m$ = strong penalty), and for DP across the
concentration $\alpha$. Two features are worth emphasizing. First, the overall
ATT is remarkably robust to over-pooling, where its bias stays small
($\lesssim 0.015$) across the entire $\ell_0$ path down to $m=3$ and only spikes
once the path collapses toward a single group ($|$bias$|=0.10$ at $m=1$). The aggregate thus tolerates a wide
range of penalties, whereas the cohort-time effects require getting the partition
approximately right, the CATT variance is minimized near the true $m=6$ and
inflates when the path is forced into wrong merges. Second, the DP path is nearly
flat in $\alpha$: both the ATT bias and variance barely move over
$\alpha\in\{1,\dots,30\}$, so the practitioner's choice of concentration has
little effect on the reported estimates. The mean BIC-selected number of groups
($\approx 6.1$) sits in the flat, low-bias region of the $\ell_0$ path, confirming
that the plateau in $\hat m$ is a serviceable tuning diagnostic.

\begin{figure}[H]
  \centering
  \includegraphics[width=\linewidth]{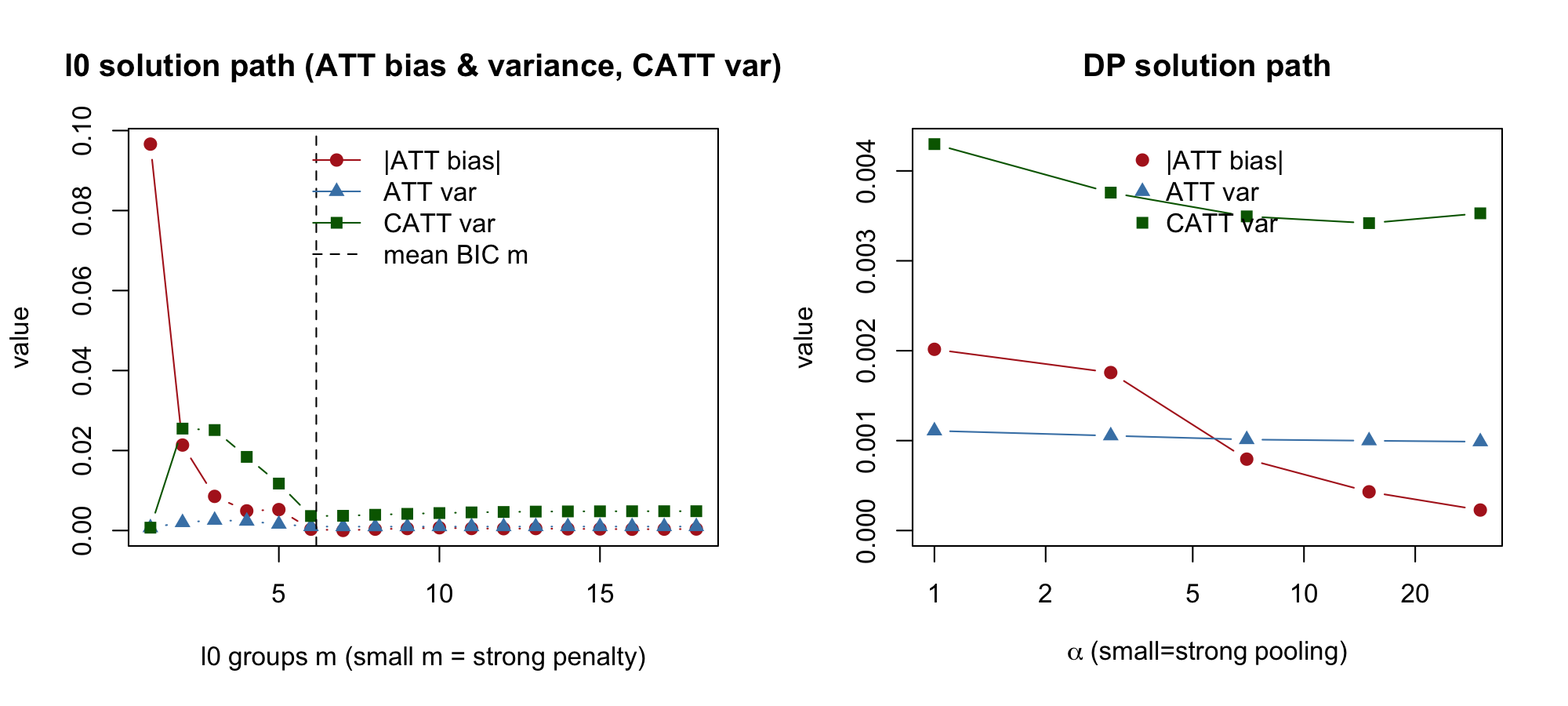}
  \caption{Solution paths at $m^{\ast}=6$, $\delta=6$. Left: $\ell_0$, indexed by
  the number of groups $m$ (dashed line = mean BIC selection). Right: DP, indexed
  by the concentration $\alpha$ (log scale). The overall ATT bias (red) stays
  near zero across a wide range, while the CATT variance (green) is non-monotone: it dips at the true
  partition ($m=6$), inflates as further pooling forces wrong merges, and drops again only at full
  pooling ($m=1$), where the bias spikes.}
  \label{fig:sim:path}
\end{figure}

\section{Empirical Applications}
\label{sec:empirical_application}

We present two applications that exhibit the two regimes the method is built for, one where the cohort-time effects carry recoverable heterogeneity and one where they are indistinguishable from a common value.

The simulations generate the micro panel directly, but the natural inputs in an application are a vector of first-stage cohort-time effects and their joint sampling covariance, produced by whichever heterogeneity-robust estimator suits the design, on which the partition model then operates. Write $\hat\tau$ for the fully flexible estimator of the $K$ cohort-time effects and $\hat\Sigma$ for its covariance, and collect a partition $\mathcal{P}=\{C_1,\dots,C_m\}$ into the $K\times m$ group-indicator matrix $R$, so that partial homogeneity is the restriction $\tau=R\phi$. The first-stage estimator is asymptotically Gaussian, which gives the two-stage likelihood
\begin{equation}
  \hat\tau \mid \phi,\mathcal{P} \;\sim\; N\!\left(R\phi,\ \hat\Sigma\right),
  \label{eq:app:twostage}
\end{equation}
the model of Section~\ref{sec:bayes:model} read in the coordinates of the first-stage estimates: the group effects keep the conjugate prior \eqref{eq:bayes:phi} and integrate out as before. Nothing downstream depends on the first stage beyond the pair $(\hat\tau,\hat\Sigma)$.

\begin{lemma}[First-stage summaries are sufficient]
\label{lem:sufficiency}
For any partition $\mathcal{P}$ with group-indicator matrix $R$, the collapsed marginal likelihood \eqref{eq:bayes:logmarg} and the group-effect posterior \eqref{eq:bayes:phipost} depend on the data only through $R'\hat\Sigma^{-1}\hat\tau$ and $R'\hat\Sigma^{-1}R$. Hence the within-transformed model of Section~\ref{sec:bayes:model} and the summary-statistic model \eqref{eq:app:twostage} induce the same posterior, and any consistent first-stage covariance $\hat\Sigma$ enters in place of $\sigma^{2}(\Dtd'\Dtd)^{-1}$.
\end{lemma}

\begin{proof}
The flexible estimator satisfies $\Dtd'\Omega^{-1}\Dtd=\hat\Sigma^{-1}$ and $\Dtd'\Omega^{-1}\tilde{Y}=\hat\Sigma^{-1}\hat\tau$ for the first-stage error covariance $\Omega$. Under $\mathcal{P}$ the grouped design is $\Dtd^{*}=\Dtd R$, so the two moments that carry the data into \eqref{eq:bayes:logmarg} and \eqref{eq:bayes:phipost} (Appendix~\ref{app:hetero}, eq.~\eqref{eq:het:moments} with $\Omega$ the first-stage covariance) are $\Dtd^{*\prime}\Omega^{-1}\Dtd^{*}=R'\hat\Sigma^{-1}R$ and $\Dtd^{*\prime}\Omega^{-1}\tilde{Y}=R'\hat\Sigma^{-1}\hat\tau$. Taking $\hat\Sigma=\sigma^{2}(\Dtd'\Dtd)^{-1}$ returns the homoskedastic moments of Section~\ref{sec:bayes:model}.
\end{proof}

The covariance $\hat\Sigma$ is thus the only channel through which the design enters, and in each application we take it from the first-stage estimator itself, so the grouping is applied to exactly the covariance the design implies.

\subsection{Minimum Wages and Teen Employment: Recovering Partial Homogeneity}
\label{sec:cs_application}

We revisit the county-level analysis of minimum wages and teen employment from
\citet{callaway2021}, using their publicly available panel of $500$ U.S.\
counties observed over $2003$--$2007$, with cohorts first raising the minimum
wage in $2004$, $2006$, and $2007$ and a not-yet-treated control group. We
estimate the $K=7$ post-treatment cohort-time effects $\mathrm{ATT}(g,t)$ with the
\citet{callaway2021} estimator and take their exact joint covariance
$\hat\Sigma$ from the estimator's influence functions, the pair $(\hat\tau,\hat\Sigma)$
of Lemma~\ref{lem:sufficiency}. Because the cells share
counties, $\hat\Sigma$ is non-diagonal; as in the simulations, we use it exactly
in both estimators, and the DP posterior draws use the corresponding
Gauss--Markov posterior (Remark~\ref{rem:sigma}).

Unlike a homogeneous design, here there is real structure to recover: the
cross-cell dispersion of the estimates ($0.050$) is more than twice the typical
standard error ($0.022$), so the spread exceeds what sampling noise alone would
produce and, unlike the homogeneous case of Section~\ref{sec:cengiz_application},
reflects genuine heterogeneity. This descriptive signal-to-noise ratio is not the
separation $\delta$ of Section~\ref{sec:sim}, which is defined on the true
adjacent-group gaps and the full joint covariance; it signals recoverable
heterogeneity without fixing its exact position on the $\delta$ scale, which the
co-clustering results below assess directly.
Table~\ref{tab:cs} reports the seven effects, the $\ell_0$ grouping selected by
BIC, and the per-cell variance reduction.

\begin{table}[H]
    \centering
    \caption{Cohort-time effects on log teen employment, \citet{callaway2021}
    data. $\mathrm{ATT}(g,t)$ and its standard error; the BIC-selected $\ell_0$
    group and grouped effect; and the ratio of the $\ell_0$-grouped variance to
    the flexible variance for each cell (lower = more precise).}
    \label{tab:cs}
    \small
    \begin{tabular}{lccccc}
        \toprule
        cohort:year & $\mathrm{ATT}(g,t)$ & SE & $\ell_0$ group & grouped effect & var.\ ratio \\
        \midrule
        2004:2004 & $-0.019$ & 0.022 & 1 & $-0.029$ & 0.28 \\
        2006:2007 & $-0.041$ & 0.020 & 1 & $-0.029$ & 0.33 \\
        2007:2007 & $-0.026$ & 0.017 & 1 & $-0.029$ & 0.49 \\
        2004:2005 & $-0.078$ & 0.030 & 2 & $-0.087$ & 0.77 \\
        2004:2007 & $-0.101$ & 0.034 & 2 & $-0.087$ & 0.61 \\
        2004:2006 & $-0.136$ & 0.035 & 3 & $-0.140$ & 0.78 \\
        2006:2006 & $+0.005$ & 0.016 & 4 & $+0.011$ & 0.78 \\
        \bottomrule
    \end{tabular}
\end{table}

The $\ell_0$-PH estimator collapses the seven cells into four effect levels: a
near-zero group containing the $2004$ impact effect and the recent cohorts'
effects, the $2004$ cohort's larger mature effects, its peak, and the lone
positive cell. Pooling roughly halves the variance of the cells that share a
group (the near-zero group's variance ratios are $0.28$--$0.49$), and even the
singleton cells gain precision because the Gauss--Markov estimator borrows across
the correlated cells. The overall (cohort-size-weighted) average effect is
essentially unchanged and slightly more precise, $-0.040$ (SE $0.012$) under
the flexible estimator versus $-0.039$ (SE $0.012$) under $\ell_0$, so the
substantive conclusion of a roughly $4\%$ teen-employment reduction is preserved
while the estimates are summarized by a handful of interpretable effect levels.
These point estimates agree on the headline, but both condition on a single
grouping, and the Dirichlet-Process posterior that follows both gauges how firmly
that grouping is identified and, given the small number of cells, shows how much
the overall summary can move once that conditioning is relaxed.

The DP posterior sharpens the interpretation by quantifying which groupings
are certain. Figure~\ref{fig:cs} plots the effects colored by $\ell_0$ group and
the posterior co-clustering matrix, formed with the exact within-cell covariance
at the parsimonious prior $\alpha=1$.
The small-effect cells recur together with moderate-to-high probability, the
$2004$ impact effect co-clustering with the lone positive $2006{:}2006$ cell at
$0.86$ and with the $2004{:}2005$ effect at $0.77$, so a low-effect cluster is a
robust feature. The large mature $2004$ effects are more distinct, with the peak
$2004{:}2006$ co-clustering with every other cell at $0.54$ or below. The matrix
is more diffuse than any single partition, correctly signaling that with only
seven correlated cells the finer grouping is genuinely uncertain rather than a
firm partition. This is the honest reading the co-clustering is meant to deliver,
a recurring coarse structure with calibrated uncertainty about the finer detail.
Because the co-clustering probabilities are conditional on $\alpha$, we checked
their sensitivity: across $\alpha\in[0.5,14]$ the near-zero cells co-cluster more
tightly with one another (mean pairwise probability falling from $0.66$ to $0.32$
as larger $\alpha$ favors finer partitions) than the peak $2004{:}2006$ cell does
with any other cell ($0.63$ down to $0.22$). The absolute probabilities decline
with $\alpha$, but the ordering---a cohesive low-effect group and a distinct
peak---persists throughout, so the coarse structure is not an artifact of the
prior.

\begin{figure}[H]
    \centering
    \begin{subfigure}{0.54\linewidth}\includegraphics[width=\linewidth]{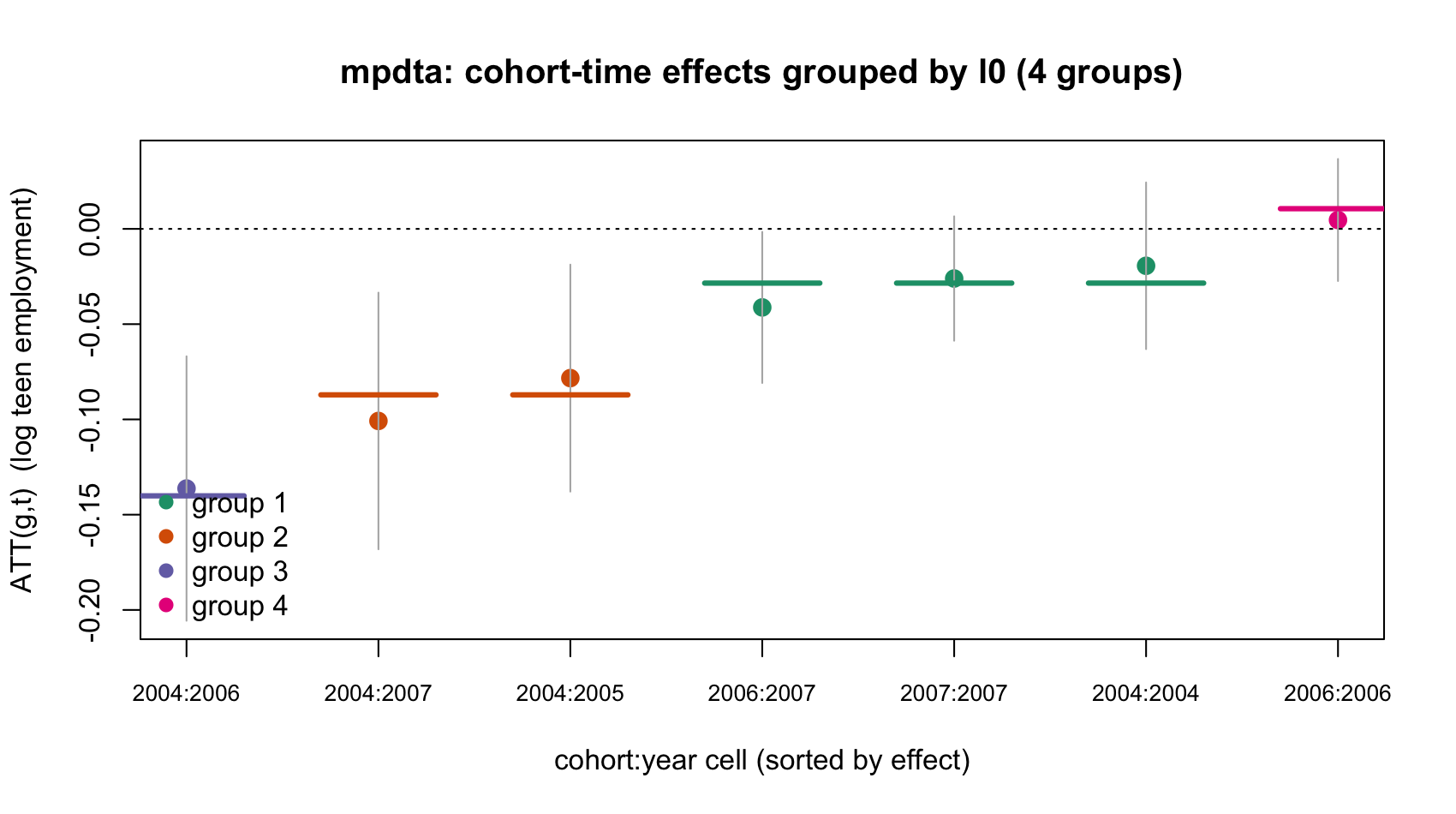}\end{subfigure}
    \hfill
    \begin{subfigure}{0.44\linewidth}\includegraphics[width=\linewidth]{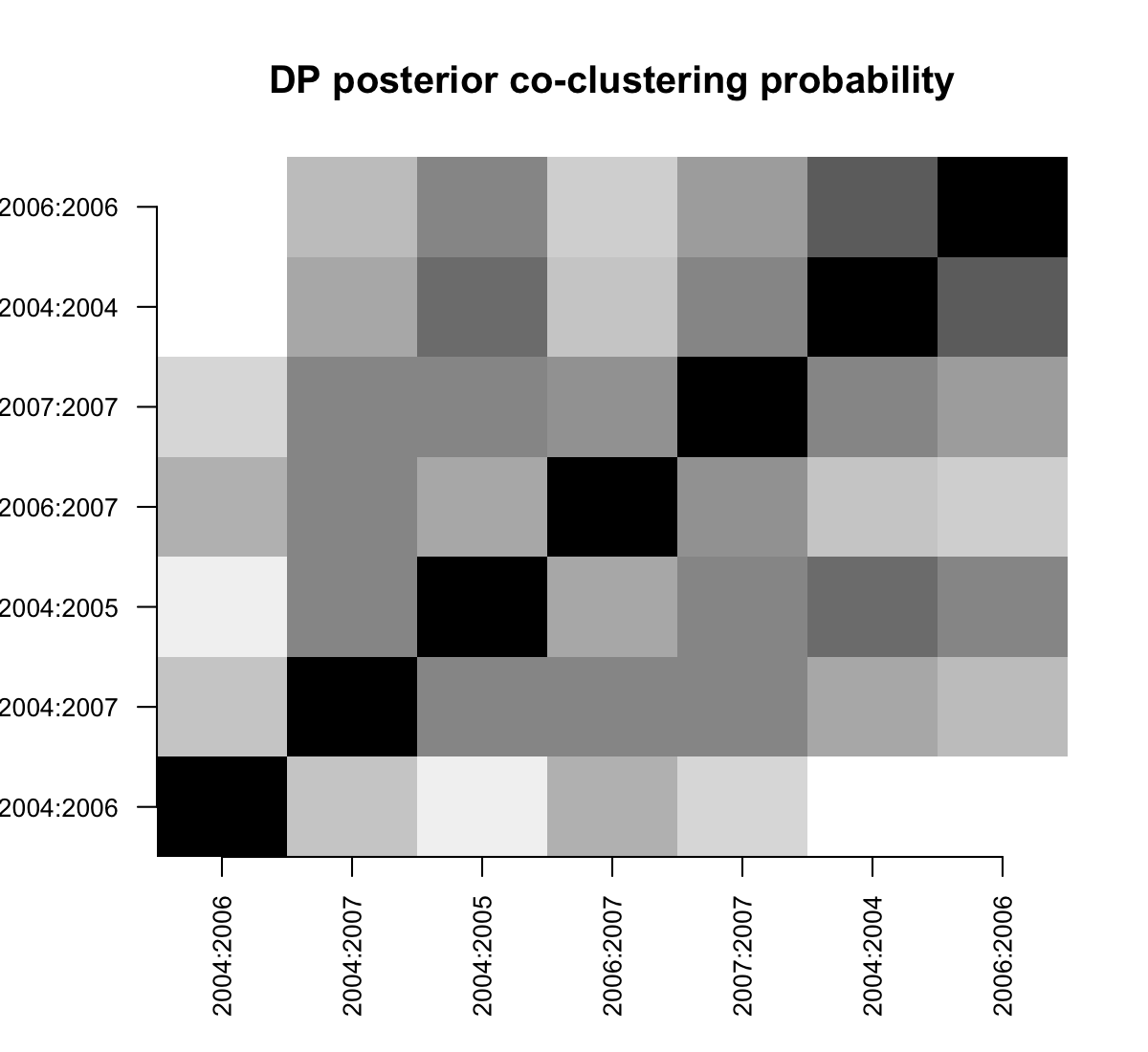}\end{subfigure}
    \caption{Left: the seven cohort-time effects (with $95\%$ intervals) colored
    by $\ell_0$ group, horizontal bars mark the grouped effects. Right: the DP
    posterior co-clustering probabilities; darker cells co-cluster more often.}
    \label{fig:cs}
\end{figure}

Because the number of cells is small, the DP overall effect, unlike the flexible and $\ell_0$-PH point estimates, is sensitive to the concentration $\alpha$. Figure~\ref{fig:cs_alpha} traces the population-weighted effect and its $95\%$ credible band across $\alpha$ from $0.1$ to $100$, against the fully pooled ($-0.010$) and flexible ($-0.040$) anchors. The effect slides monotonically from about $-0.017$ under heavy pooling to about $-0.038$ as $\alpha$ grows and the model approaches the flexible fit, and it is bounded away from zero only for $\alpha \gtrsim 1$. Rather than fix $\alpha$, we report the whole path. As a reference point, a BIC pass on the agglomeration path selects four groups, which the DP's posterior expected number of groups matches at $\alpha \approx 14$ (panel b), where the DP effect is about $-0.032$ and close to the flexible estimate. We offer this only as a suggestion. Choosing $\alpha$ from the data turns the prior into a data-dependent object, with the usual consequences of empirical Bayes (understated posterior uncertainty and a double use of the data), so we prefer to show the sensitivity rather than commit to a single value.

\begin{figure}[H]
    \centering
    \includegraphics[width=0.8\linewidth]{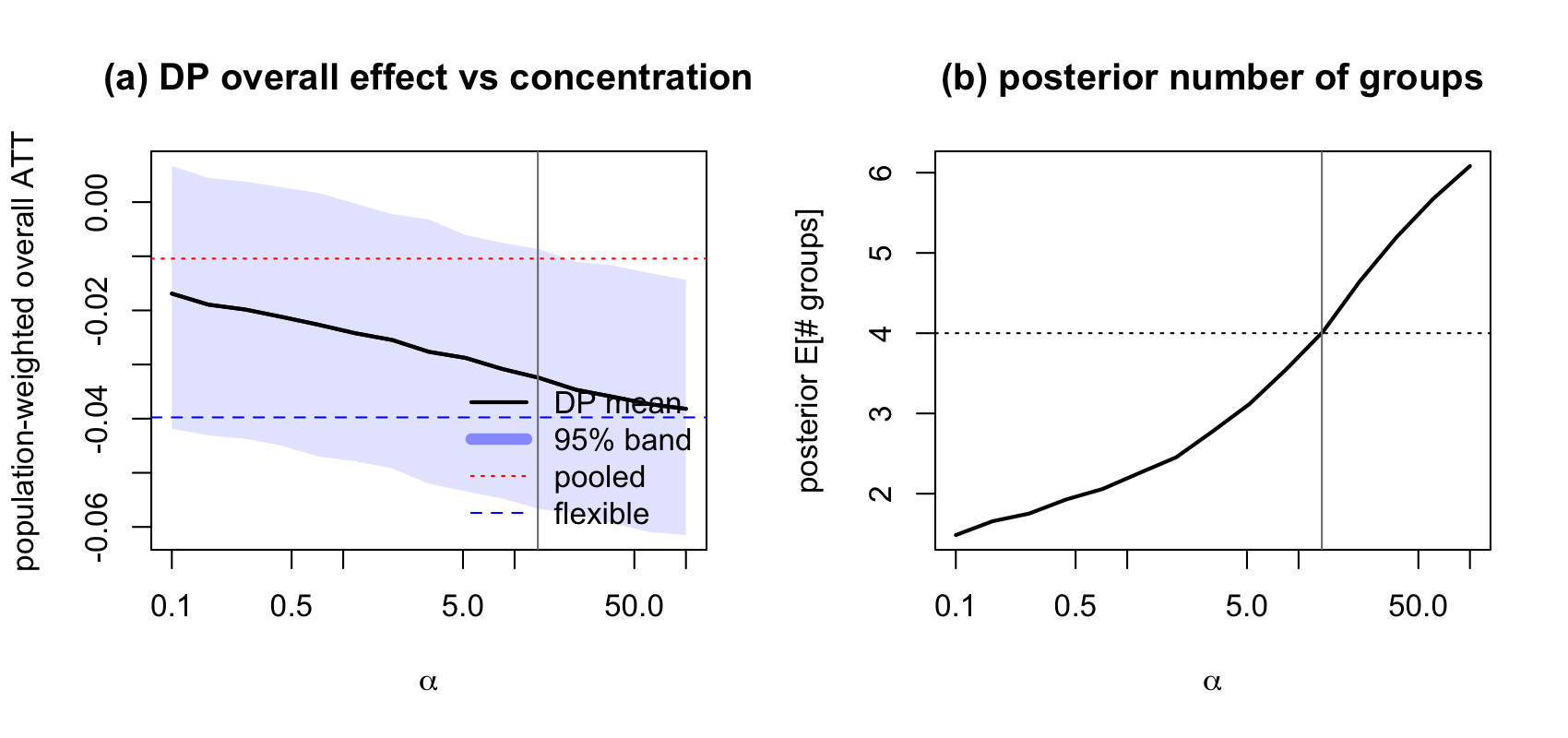}
    \caption{Sensitivity of the \citet{callaway2021} overall effect to the DP
    concentration $\alpha$. Panel (a), the population-weighted overall ATT (solid)
    with its $95\%$ credible band (shaded), against the fully pooled (red dotted,
    $-0.010$) and flexible (blue dashed, $-0.040$) benchmarks. Panel (b), the
    posterior expected number of groups. The vertical line marks $\alpha \approx 14$,
    where the posterior expected number of groups equals the four groups selected
    by BIC on the agglomeration path.}
    \label{fig:cs_alpha}
\end{figure}

The dynamic (event-study) aggregation in Figure~\ref{fig:cs_dyn} clarifies what
the estimators are grouping. The effect on teen employment is small at impact
(event time $0$, $-0.019$) and grows over the horizon to $-0.054$, $-0.136$, and
$-0.101$ at event times $1$ through $3$, the familiar pattern of a
minimum-wage effect that accumulates over several years. This accumulation is,
however, a within-cohort reading: the composition of the event-time aggregate
changes across the horizon, and event times $2$ and $3$ are identified off the
$2004$ cohort alone, so the longer-run trajectory is that cohort's and the design
cannot say whether later cohorts would follow it. The impact-year effects, by
contrast, are small and similar across cohorts. The
partial homogeneity the method recovers is, in this application, largely the
homogeneity of the short-run response across cohorts, with the accumulating
longer-run effects, carried by the $2004$ cohort's later cells, kept separate. The pre-treatment coefficients are small
(event times $-2$ and $-1$), with a modest positive value at the longest lead
that is identified off a single early cohort.

\begin{figure}[H]
    \centering
    \includegraphics[width=0.62\linewidth]{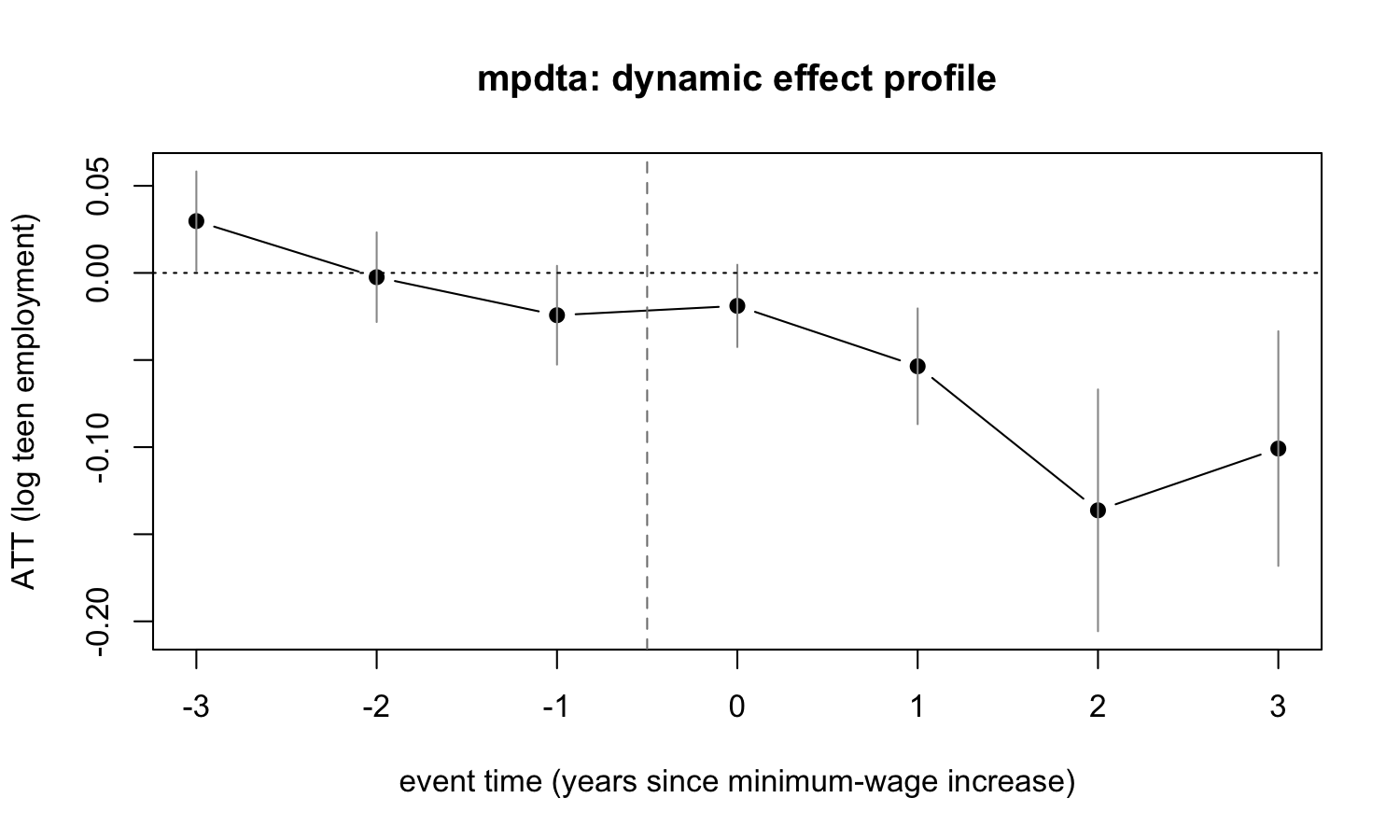}
    \caption{Dynamic (event-study) aggregation of the \citet{callaway2021}
    effects: the teen-employment effect is near zero at impact and builds over
    event time. Pre-treatment coefficients (left of the dashed line) are small.}
    \label{fig:cs_dyn}
\end{figure}

\subsection{Minimum Wages and Low-Wage Jobs: A Specification Test}
\label{sec:cengiz_application}

Our second application is the influential study of minimum wages and low-wage
jobs by \citet{cengiz2019}, which estimates the employment effect of $138$
state-level minimum wage increases between $1979$ and $2016$ using a ``stacked''
event-study design. Here the method plays the complementary role of a
specification test.

\subsubsection{Context and Stacked Event-Study Design}

In the stacked design, a separate dataset (a ``stack'' or ``event'') is created for each of the $E = 138$ minimum wage hikes. Each stack $g \in \{1, \dots, E\}$ contains the treated state and a set of clean control states observed over a 5-year window around the policy implementation. The stacks are then appended together.

The authors' baseline specification is a fully pooled TWFE regression on the stacked data, which imposes a single, homogeneous treatment effect across all 138 events:
\begin{equation} \label{eq:pooled_cengiz}
    Y_{itg} = \alpha_{ig} + \gamma_{tg} + \theta \cdot D_{itg} + \epsilon_{itg},
\end{equation}
where $Y_{itg}$ is the employment outcome for state $i$ at time $t$ in stack $g$. The terms $\alpha_{ig}$ and $\gamma_{tg}$ are state-by-stack and time-by-stack fixed effects, ensuring that treated units are only compared to control units within their specific sub-experiment. $D_{itg}$ is an indicator equal to 1 if state $i$ is treated in stack $g$ in the post-treatment period. Under this specification, the authors report a near-zero average effect on net low-wage employment.

However, as established in Section~\ref{sec:problem}, under heterogeneity this pooled estimator $\hat{\theta}$ is efficient but biased for the individual event effects $\tau_g$ and for any aggregate weighted differently from the design; it remains unbiased for its own design-weighted average, so ``bias'' here is relative to the disaggregated target rather than to that aggregate. To allow for heterogeneity, we relax Equation~\ref{eq:pooled_cengiz} into a fully flexible Interacted TWFE (ITWFE) model---the stacked-design counterpart of the flexible specification \eqref{eq:twfe_flex} of Section~\ref{sec:problem}---assigning a distinct Cohort-Average Treatment Effect on the Treated (CATT), $\tau_g$, to each event:
\begin{equation} \label{eq:flexible_cengiz}
    Y_{itg} = \alpha_{ig} + \gamma_{tg} + \sum_{e=1}^{E} \tau_e \left( D_{itg} \times \mathbb{I}\{event = e\} \right) + \epsilon_{itg}
\end{equation}

Because the design matrix of Equation~\ref{eq:flexible_cengiz} is block-diagonal, the estimation of each $\hat{\tau}_g$ is perfectly separable. Block-diagonality makes the point estimates separable but not their sampling errors, since control states recur across stacks, so the event estimates need not be independent. Our homogeneity assessment is robust to this: the within-event randomization test below holds the cross-event dependence fixed by construction, and the pre-trend noise reference enters the signal-to-noise comparison symmetrically with the post-treatment spread, so a common dependence largely cancels. The grouping and the aggregate credible interval, by contrast, use a diagonal covariance and would be sharpened by a joint state-clustered $\hat\Sigma$, which enters the estimators unchanged through Lemma~\ref{lem:sufficiency} but requires the stacked micro panel rather than the event-level summaries in the replication package; because the events carry no recoverable heterogeneity, this refinement does not alter the substantive conclusion. While this flexible model is unbiased, relying strictly on state-level variation yields highly noisy estimates for individual events.

To recover the latent structure, we apply our $\ell_0$-penalized estimator, which selects a partition of the 138 minimum wage effects by minimizing the penalized residual sum of squares:
\begin{equation} \label{eq:penalized_cengiz}
    \min_{\boldsymbol{\tau}, \boldsymbol{\alpha}, \boldsymbol{\gamma}} \left\{ \sum_{g=1}^{E} \sum_{i,t} \left( Y_{itg} - \alpha_{ig} - \gamma_{tg} - D_{itg} \tau_g \right)^2 + \lambda \sum_{j < k} \mathbb{I}(\tau_j \neq \tau_k) \right\},
\end{equation}
Exact minimization over partitions is intractable, so the greedy agglomerative search of Appendix~\ref{app:algorithm} returns the selected partition; $\lambda$ is the hyperparameter determining the threshold for model parsimony. The penalty forces states with sufficiently similar independent estimates to share a single coefficient, achieving partial homogeneity without imposing the strict global restriction of Equation~\ref{eq:pooled_cengiz}.

\subsubsection{Baseline Replication and Data Hygiene}

Before applying the penalized estimator to explore latent heterogeneity, we first replicate the fully pooled TWFE specification from \citet{cengiz2019} to establish a baseline. Table~\ref{tab:baseline_replication} demonstrates that our replication sample accurately recovers the primary pooled estimates reported by the authors. The pooled model implies a near-zero average net effect on low-wage employment, driven by the destruction of jobs below the new minimum wage ($\Delta b$) being largely offset by the ``bunching'' of new compliance jobs at or slightly above the new minimum wage ($\Delta a$).

\begin{table}[H]
    \centering
    \caption{Baseline Pooled ITWFE Replication}
    \label{tab:baseline_replication}
    \begin{tabular}{lc}
        \toprule
        \textbf{Employment Outcome} & \textbf{Pooled Estimate} \\
        \midrule
        Missing Jobs ($\Delta b$) & $-0.016^{***}$ \\
        Excess Jobs ($\Delta a$) & $0.019^{***}$ \\
        Net Employment Change & $0.003$ \\
        \bottomrule
    \end{tabular}
    \vspace{0.1cm}
    \begin{minipage}{0.7\textwidth}
        \small Note: Estimates represent the change in employment as a fraction of the affected workforce. These pooled point estimates closely mirror the benchmark findings of the original study, suggesting a negligible average disemployment effect.
    \end{minipage}
\end{table}

While the pooled model provides a stable average, estimating Equation~\ref{eq:flexible_cengiz} flexibly for all 138 independent events reveals severe localized sampling anomalies. Because the employment effects are normalized by $\bar{b}_{-1}$ (the share of the state's workforce earning below the new minimum wage in the year prior to the hike), regions with very thin survey samples can produce a denominator approaching zero. Dividing the raw job counts by a near-zero denominator causes the resulting percentage estimate to artificially explode to economically impossible magnitudes.

A handful of events yield mathematically anomalous net percentage effects (exceeding several hundred percent of the affected workforce) for exactly this reason. To sidestep the near-zero-denominator problem, our specification analysis below summarizes each event by its average post-treatment event-study coefficient on affected employment rather than by the ratio-normalized percentage effect; the event-study coefficient is well behaved (its cross-event standard deviation is $0.005$, with no explosive outliers). This substitution changes the analyzed object from the $\bar b_{-1}$-normalized percentage effect to the response on the affected-employment scale. Because $\bar b_{-1}$ varies across events, homogeneity and aggregation need not transfer between the two scales; the specification analysis below therefore concerns the affected-employment responses, the object that carries the employment signal and the natural scale on which to aggregate a total-employment effect, and complements rather than re-derives the normalized headline of \citet{cengiz2019}.

\subsubsection{Specification Analysis: Is the Null Effect Homogeneous?}
\label{sec:cengiz:spec}

We obtain the $E = 138$ event-level effects from the authors' Harvard Dataverse replication package for \citet{cengiz2019}, summarizing event $g$ by $\hat\tau_g$, the average of its post-treatment event-study coefficients on affected employment (event times $1$ through $4$), constructed symmetrically with the four pre-treatment coefficients (event times $-4$ through $-1$). Before selecting a specification, we ask whether there is any heterogeneity to recover. The design provides an internal noise gauge: under no-anticipation the pre-treatment coefficients have expectation zero, and because each event's pre- and post-treatment summaries average the same number of event-study coefficients from the same design, the cross-event dispersion of the pre-treatment averages estimates the sampling variance of the post-treatment averages under the null of a common effect. We find a pre-treatment (noise) standard deviation of $s_{\text{pre}}=0.0065$, which exceeds the post-treatment cross-event spread of $s_{\text{post}}=0.0050$. Decomposing the post-treatment cross-event variance into sampling noise and genuine heterogeneity, $\hat\sigma^2_{\text{het}}=\max\{0,\,s^2_{\text{post}}-s^2_{\text{pre}}\}$, yields $\hat\sigma^2_{\text{het}}=0$ (a $0\%$ heterogeneity share) at a signal-to-noise ratio $s_{\text{post}}/s_{\text{pre}}=0.77$, so the 138 event effects are statistically indistinguishable from a common value. A within-event randomization test makes this precise: under the null of a common effect we exchange each event's pre- and post-treatment averages, an exchange internal to each event that holds the cross-event dependence structure fixed, and compare the observed post-treatment dispersion to the resulting reference distribution. It sits below the reference median (one-sided $p=0.71$), so a common effect cannot be rejected.

This places the application squarely in the low-separation regime that Section~\ref{sec:sim} identifies as the boundary of reliable recovery, and the estimators behave accordingly. The BIC-tuned $\ell_0$-PH estimator splits the events into five effect-ordered bands (of sizes $51$, $47$, $36$, $3$, $1$), and the Bayes-PH estimator returns between $2.7$ groups (parsimonious prior, $\alpha=1$) and $7.6$ groups ($\alpha=5$). As our simulations show, at a signal-to-noise ratio below one these ``groups'' are quantiles of noise rather than genuine effect clusters, and should not be interpreted as evidence of systematic heterogeneity; indeed, the randomization test above does not reject a common effect. The disagreement between the two estimators is itself a symptom of the regime. None of this proves an exactly common effect, since a null of equality cannot be accepted; but the failure to reject, together with the zero estimated heterogeneity share, supports pooling as an adequate specification for the aggregate rather than establishing a literally homogeneous set of effects.

What is robust, and what matters for the applied conclusion, is that the overall average effect is invariant to the specification. Table~\ref{tab:cengiz:att} reports it, showing that the pooled, $\ell_0$, and Bayes-PH estimates of the population-weighted average of the affected-employment responses are all near zero and mutually indistinguishable, and the DP posterior credible interval is a tight band of small positive values, economically negligible and far from any meaningful disemployment effect; its width, which conditions on cross-event independence, is not the basis for any conclusion here. Allowing for arbitrary event-level heterogeneity in these responses creates no recoverable structure for a pooled specification to miss, so on the affected-employment scale pooling is adequate. This is stated on the coefficient scale of Table~\ref{tab:cengiz:att} rather than the $\bar b_{-1}$-normalized scale of Table~\ref{tab:baseline_replication},\footnote{The two aggregates are on different scales and are not directly comparable. Table~\ref{tab:cengiz:att} averages the affected-employment event-study coefficients $\hat\tau_g$, whereas Table~\ref{tab:baseline_replication} reports the $\bar b_{-1}$-normalized pooled effect; per event the two are related by $\tau_g^{\mathrm{pct}}=\hat\tau_g/\bar b_{-1,g}$, and the near-zero denominators in thin cells are what make the per-event normalized effects explode (here from roughly $-130\%$ to $+420\%$) and motivate the coefficient-scale summary. Neither number is a simple rescaling of the other, since the normalized benchmark is a single pooled coefficient rather than an average of the per-event ratios; both are small and positive.} and so complements, rather than re-derives, the near-zero normalized net-employment finding of \citet{cengiz2019}. Figure~\ref{fig:cengiz} displays the 138 flexible event effects with the pooled effect superimposed.

\begin{table}[H]
    \centering
    \caption{Overall employment effect under alternative specifications, Cengiz et al.\ (2019) events. Effects are average post-treatment coefficients on affected employment (population-weighted for the aggregate). The Bayes-PH interval is the $95\%$ posterior credible interval.}
    \label{tab:cengiz:att}
    \small
    \begin{tabular}{lcc}
        \toprule
        Specification & Groups & Overall effect \\
        \midrule
        Pooled (single effect)      & $1$    & $0.0015$ \\
        Flexible (per event)        & $138$  & $0.0018$ \\
        $\ell_0$-PH (BIC)              & $5$    & $0.0016$ \\
        Bayes-PH ($\alpha=1$)       & $2.7$  & $0.0014$\ \ $[0.0003,\,0.0026]$ \\
        \bottomrule
    \end{tabular}
\end{table}

\begin{figure}[H]
    \centering
    \includegraphics[width=0.8\linewidth]{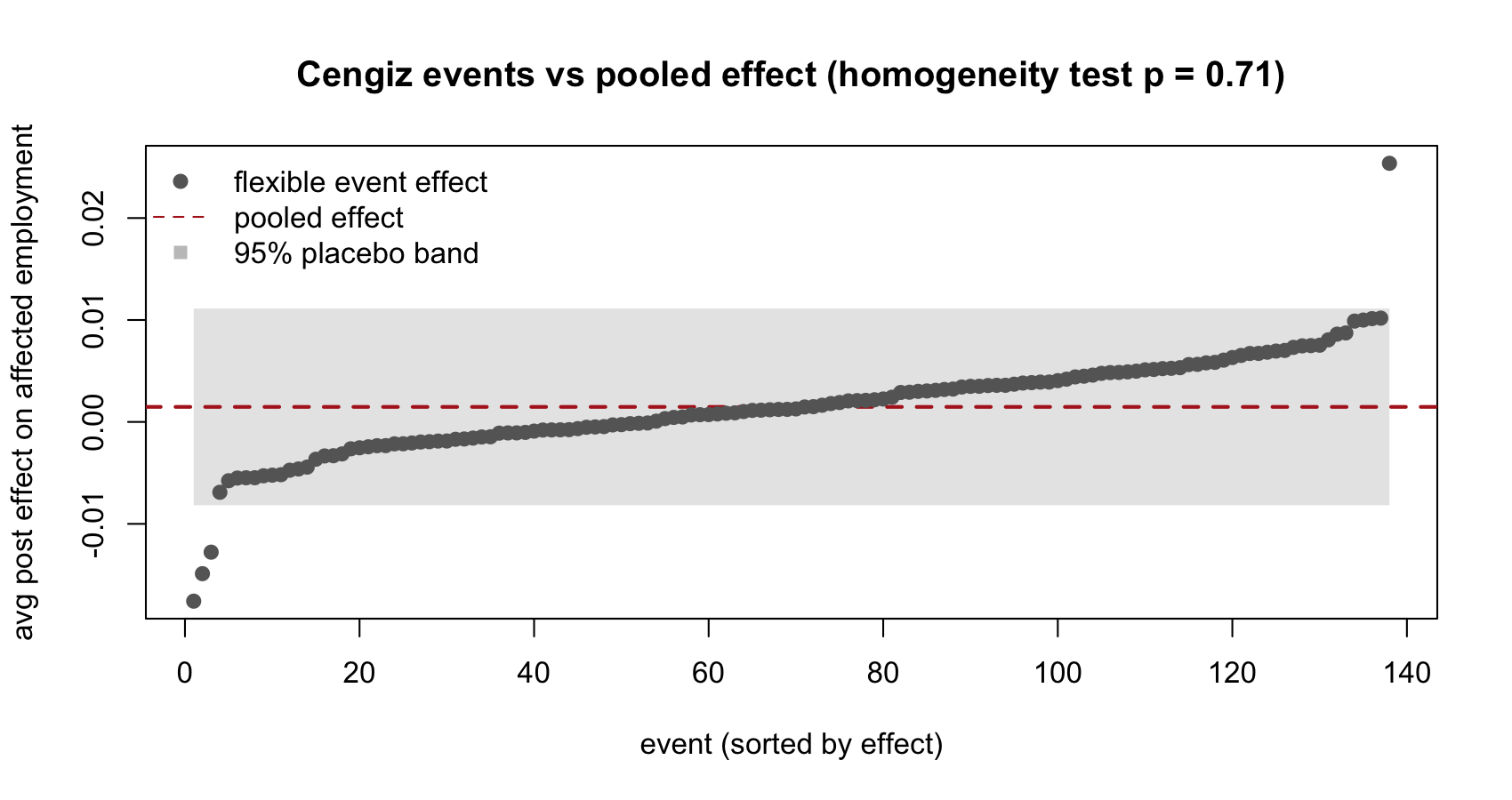}
    \caption{The 138 event-level employment effects (average post-treatment coefficient on affected employment), sorted; the dashed line is the pooled effect and the shaded region is the $95\%$ placebo band, the spread of the pre-treatment coefficients expected under a common effect. The cross-event dispersion is no larger than this noise band (within-event randomization test $p=0.71$), and the pooled effect is near zero.}
    \label{fig:cengiz}
\end{figure}

Finally, we note that \citet{cengiz2019} themselves group events into three ``Card--Krueger'' bins defined ex ante by how binding the minimum wage is. Our exercise asks, and answers in the negative, a different question, whether the employment effects, taken as estimated, cluster into distinct levels. The two are complementary, in that grouping by a policy covariate is a modeling choice about expected heterogeneity, whereas our data-driven partition tests for realized heterogeneity in the effects and, here, finds none beyond noise.

\subsubsection{Diagnosing Canonical TWFE Bias: The Goodman-Bacon Decomposition}

To motivate the stacked design, we decompose the canonical (unstacked) two-way fixed effects estimand following \citet{goodman2021}. Two adjustments are needed to apply the decomposition here, since we proxy the affected workforce by employment below a fixed \$10 threshold (the true minimum-wage threshold moves across states and years), and we define treatment by the first increase a state enacts (to obtain an absorbing treatment). Both are conservative and, if anything, understate this contamination.

Table~\ref{tab:bacon_decomp} reports the result. About $20\%$ of the estimand's weight comes from ``forbidden'' comparisons that use already-treated states as controls, carrying a negative average estimate ($-0.006$), against the positive ($0.006$) of the clean treated-versus-untreated comparisons (about $57\%$ of the weight). The pooled coefficient is the weighted average of these components, so the forbidden comparisons enter it negatively and pull it below the clean comparisons; since the clean and forbidden blocks weight different cohorts and exposure horizons, what the decomposition exposes is contamination of the pooled estimand rather than a signed bias against a fixed target. This contamination is the standard rationale for the stacked event-study design used above, and more broadly for taking the cohort-time effects, rather than a single pooled coefficient, as the objects of interest that the specification methods of this paper then group.

\begin{table}[H]
    \centering
    \caption{Goodman--Bacon decomposition of the canonical TWFE estimand (simplified absorbing-treatment panel).}
    \label{tab:bacon_decomp}
    \begin{tabular}{lcc}
        \toprule
        Comparison type & Weight & Average estimate \\
        \midrule
        Earlier vs.\ later treated & 0.236 & $\phantom{-}0.001$ \\
        Later vs.\ earlier treated & 0.195 & $-0.006$ \\
        Treated vs.\ untreated     & 0.569 & $\phantom{-}0.006$ \\
        \bottomrule
    \end{tabular}
\end{table}

\section{Conclusion}
\label{sec:conc}

This paper addresses a step in the staggered DiD workflow that has received
comparatively little attention, model specification. Once cohort-time effects are
identified as the estimands of interest, the researcher must still decide how many
distinct parameters to estimate. We frame this as a partition-selection problem
and address it with a Dirichlet Process mixture over the cohort-time effects. The
DP posterior assigns a probability to every grouping and reports full uncertainty,
including the uncertainty in the grouping itself, and is the object we recommend
for inference. Holding the error variance fixed and adopting a pairwise penalty
on the partition, the maximum a posteriori estimator coincides with an
$\ell_0$-penalized estimator, a connection to
the homogeneity-pursuit literature. As a secondary benefit, the
partial-homogeneity restriction can also supply identifying information for a cell
that lacks a clean comparison, though that identification rests entirely on the
assumed grouping and is untestable for the cell that needs it.

The simulations and applications make the trade-offs concrete. Under partial
homogeneity with well-separated effect levels, both estimators cut the sampling
variance of the cohort-time effects by $26$--$52\%$ relative to the fully
flexible estimator at no bias cost, approaching the infeasible oracle as the
separation grows, with the advantage narrowing and, once the effects are too close to separate
reliably, reversing into a precision loss as the feasible estimators fall behind
flexible TWFE, and the DP posterior
attains near-nominal coverage by marginalizing over the partition where naive
$\ell_0$-PH intervals are usable but mildly anti-conservative. The two
applications show both regimes on real data. In the minimum-wage and
teen-employment data of \citet{callaway2021} the effects are genuinely
heterogeneous and the method recovers a partially homogeneous structure,
collapsing seven effects into a few interpretable levels and roughly halving the
variance of the pooled cells while preserving the headline effect, whereas in the
low-wage-jobs study of \citet{cengiz2019} the $138$ event-level
affected-employment responses are statistically indistinguishable from a common
value and the method correctly reports that pooling is adequate on that scale,
functioning as a formal specification test. The
same procedure thus both exploits real heterogeneity and declines to manufacture
it where none exists.

Several directions merit future work. First, the agglomerative algorithm is a heuristic; developing exact or near-exact algorithms for small $K$, or theoretical guarantees on approximation quality, would be valuable. Second, honest inference after partition selection deserves a fuller theoretical treatment, since the naive $\ell_0$-PH intervals are conditional on the selected grouping (Remark~\ref{rem:postselection}), so developing valid selective-inference corrections for them (the Bayesian posterior already propagates partition uncertainty) is an open problem. Third, formal asymptotic theory establishing consistency of the partition estimator, as both $N$ and the effective sample size per CATT grow, would solidify the method's theoretical foundations, and would formalize the separation threshold that our simulations identify as the determinant of the efficiency gain. Finally, whereas covariate adjustment in the outcome equation is immediate (Section~\ref{sec:problem}, equations~\eqref{eq:twfe_cov}--\eqref{eq:resid_cov}), letting observable characteristics inform the grouping itself, through a covariate-dependent partition prior, is a more substantial extension that could improve performance in settings with richer data.

\newpage
\bibliographystyle{aer}
\bibliography{references}

@article{goodman2021,
  author  = {Goodman-Bacon, Andrew},
  title   = {Difference-in-Differences with Variation in Treatment Timing},
  journal = {Journal of Econometrics},
  year    = {2021},
  volume  = {225},
  number  = {2},
  pages   = {254--277}
}

@article{dechaisemartin2020,
  author  = {de Chaisemartin, Cl\'ement and D'Haultf{\oe}uille, Xavier},
  title   = {Two-Way Fixed Effects Estimators with Heterogeneous Treatment Effects},
  journal = {American Economic Review},
  year    = {2020},
  volume  = {110},
  number  = {9},
  pages   = {2964--2996}
}

@article{dechaisemartin2023,
  author  = {de Chaisemartin, Cl\'ement and D'Haultf{\oe}uille, Xavier},
  title   = {Two-Way Fixed Effects and Differences-in-Differences with Heterogeneous Treatment Effects: A Survey},
  journal = {The Econometrics Journal},
  year    = {2023},
  volume  = {26},
  number  = {3},
  pages   = {C1--C30}
}

@unpublished{arora2026,
  author = {Arora, Parush and Bijani, Rishabh},
  title  = {Estimating Treatment Effects under Staggered Timing and Non-Spherical Errors},
  year   = {2026},
  note   = {Working paper, SSRN Working Paper No.\ 6558759}
}

@article{callaway2021,
  author  = {Callaway, Brantly and Sant'Anna, Pedro H. C.},
  title   = {Difference-in-Differences with Multiple Time Periods},
  journal = {Journal of Econometrics},
  year    = {2021},
  volume  = {225},
  number  = {2},
  pages   = {200--230}
}

@article{sun2021,
  author  = {Sun, Liyang and Abraham, Sarah},
  title   = {Estimating Dynamic Treatment Effects in Event Studies with Heterogeneous Treatment Effects},
  journal = {Journal of Econometrics},
  year    = {2021},
  volume  = {225},
  number  = {2},
  pages   = {175--199}
}

@article{borusyak2024,
  author  = {Borusyak, Kirill and Jaravel, Xavier and Spiess, Jann},
  title   = {Revisiting Event-Study Designs: Robust and Efficient Estimation},
  journal = {Review of Economic Studies},
  year    = {2024},
  volume  = {91},
  number  = {6},
  pages   = {3253--3285}
}

@article{gardner2022two,
  author  = {Gardner, John},
  title   = {Two-Stage Differences in Differences},
  journal = {arXiv preprint arXiv:2207.05943},
  year    = {2022}
}

@article{liu2024practical,
  author  = {Liu, Licheng and Wang, Ye and Xu, Yiqing},
  title   = {A Practical Guide to Counterfactual Estimators for Causal Inference with Time-Series Cross-Sectional Data},
  journal = {American Journal of Political Science},
  year    = {2024},
  volume  = {68},
  number  = {1},
  pages   = {160--176}
}

@article{wooldridge2025,
  author  = {Wooldridge, Jeffrey M.},
  title   = {Two-Way Fixed Effects, the Two-Way Mundlak Regression, and Difference-in-Differences Estimators},
  journal = {Empirical Economics},
  year    = {2025},
  volume  = {69},
  pages   = {2545--2587}
}

@article{su2016,
  author  = {Su, Liangjun and Shi, Zhentao and Phillips, Peter C. B.},
  title   = {Identifying Latent Structures in Panel Data},
  journal = {Econometrica},
  year    = {2016},
  volume  = {84},
  number  = {6},
  pages   = {2215--2264}
}

@article{bonhomme2015,
  author  = {Bonhomme, St\'ephane and Manresa, Elena},
  title   = {Grouped Patterns of Heterogeneity in Panel Data},
  journal = {Econometrica},
  year    = {2015},
  volume  = {83},
  number  = {3},
  pages   = {1147--1184}
}

@article{ke2015,
  author  = {Ke, Yuan and Fan, Jianqing and Wu, Yichao},
  title   = {Homogeneity Pursuit},
  journal = {Journal of the American Statistical Association},
  year    = {2015},
  volume  = {110},
  number  = {509},
  pages   = {175--194}
}

@article{kwon2026,
  author  = {Kwon, Soonwoo and Sun, Liyang},
  title   = {Estimating Treatment Effects Under Bounded Heterogeneity},
  journal = {Working paper, arXiv:2510.05454},
  year    = {2026}
}

@article{armstrong2018,
  author  = {Armstrong, Timothy B. and Koles\'ar, Michal},
  title   = {Optimal Inference in a Class of Regression Models},
  journal = {Econometrica},
  year    = {2018},
  volume  = {86},
  number  = {2},
  pages   = {655--683}
}

@article{armstrong2025,
  author  = {Armstrong, Timothy B. and Kline, Patrick and Sun, Liyang},
  title   = {Adapting to Misspecification},
  journal = {Forthcoming at Econometrica},
  year    = {2025}
}

@article{ferguson1973,
  author  = {Ferguson, Thomas S.},
  title   = {A Bayesian Analysis of Some Nonparametric Problems},
  journal = {The Annals of Statistics},
  year    = {1973},
  volume  = {1},
  number  = {2},
  pages   = {209--230}
}

@article{antoniak1974,
  author  = {Antoniak, Charles E.},
  title   = {Mixtures of Dirichlet Processes with Applications to Bayesian Nonparametric Problems},
  journal = {The Annals of Statistics},
  year    = {1974},
  volume  = {2},
  number  = {6},
  pages   = {1152--1174}
}

@article{escobar1995,
  author  = {Escobar, Michael D. and West, Mike},
  title   = {Bayesian Density Estimation and Inference Using Mixtures},
  journal = {Journal of the American Statistical Association},
  year    = {1995},
  volume  = {90},
  number  = {430},
  pages   = {577--588}
}

@article{muller1996,
  author  = {M\"uller, Peter and Erkanli, Alaattin and West, Mike},
  title   = {Bayesian Curve Fitting Using Multivariate Normal Mixtures},
  journal = {Biometrika},
  year    = {1996},
  volume  = {83},
  number  = {1},
  pages   = {67--79}
}

@article{neal2000,
  author  = {Neal, Radford M.},
  title   = {Markov Chain Sampling Methods for Dirichlet Process Mixture Models},
  journal = {Journal of Computational and Graphical Statistics},
  year    = {2000},
  volume  = {9},
  number  = {2},
  pages   = {249--265}
}

@article{muller2004,
  author  = {M\"uller, Peter and Quintana, Fernando A.},
  title   = {Nonparametric Bayesian Data Analysis},
  journal = {Statistical Science},
  year    = {2004},
  volume  = {19},
  number  = {1},
  pages   = {95--110}
}

@article{miller2018,
  author  = {Miller, Jeffrey W. and Harrison, Matthew T.},
  title   = {Mixture Models with a Prior on the Number of Components},
  journal = {Journal of the American Statistical Association},
  year    = {2018},
  volume  = {113},
  number  = {521},
  pages   = {340--356}
}

@article{hartigan1990,
  author  = {Hartigan, John A.},
  title   = {Partition Models},
  journal = {Communications in Statistics - Theory and Methods},
  year    = {1990},
  volume  = {19},
  number  = {8},
  pages   = {2745--2756}
}

@article{schwarz1978,
  author  = {Schwarz, Gideon},
  title   = {Estimating the Dimension of a Model},
  journal = {The Annals of Statistics},
  year    = {1978},
  volume  = {6},
  number  = {2},
  pages   = {461--464}
}

@article{Tibshirani2005,
  author  = {Tibshirani, Robert and Saunders, Michael and Rosset, Saharon and Zhu, Ji and Knight, Keith},
  title   = {Sparsity and Smoothness via the Fused Lasso},
  journal = {Journal of the Royal Statistical Society: Series B},
  year    = {2005},
  volume  = {67},
  number  = {1},
  pages   = {91--108}
}

@article{cengiz2019,
  author  = {Cengiz, Doruk and Dube, Arindrajit and Lindner, Attila and Zipperer, Ben},
  title   = {The Effect of Minimum Wages on Low-Wage Jobs},
  journal = {The Quarterly Journal of Economics},
  year    = {2019},
  volume  = {134},
  number  = {3},
  pages   = {1405--1454}
}

@article{wade2018,
  author  = {Wade, Sara and Ghahramani, Zoubin},
  title   = {Bayesian Cluster Analysis: Point Estimation and Credible Balls (with Discussion)},
  journal = {Bayesian Analysis},
  year    = {2018}, volume = {13}, number = {2}, pages = {559--626}
}

@article{lau2007,
  author  = {Lau, John W. and Green, Peter J.},
  title   = {Bayesian Model-Based Clustering Procedures},
  journal = {Journal of Computational and Graphical Statistics},
  year    = {2007}, volume = {16}, number = {3}, pages = {526--558}
}

@article{leeb2005,
  author  = {Leeb, Hannes and P\"otscher, Benedikt M.},
  title   = {Model Selection and Inference: Facts and Fiction},
  journal = {Econometric Theory},
  year    = {2005}, volume = {21}, number = {1}, pages = {21--59}
}

@article{rinaldo2019,
  author  = {Rinaldo, Alessandro and Wasserman, Larry and G'Sell, Max},
  title   = {Bootstrapping and Sample Splitting for High-Dimensional, Assumption-Lean Inference},
  journal = {The Annals of Statistics},
  year    = {2019}, volume = {47}, number = {6}, pages = {3438--3469}
}

@article{fithian2014,
  author  = {Fithian, William and Sun, Dennis and Taylor, Jonathan},
  title   = {Optimal Inference After Model Selection},
  journal = {arXiv:1410.2597},
  year    = {2014}
}

@article{wang2018,
  author  = {Wang, Wuyi and Phillips, Peter C. B. and Su, Liangjun},
  title   = {Homogeneity Pursuit in Panel Data Models: Theory and Application},
  journal = {Journal of Applied Econometrics},
  year    = {2018}, volume = {33}, number = {6}, pages = {797--815}
}

@article{lu2017,
  author  = {Lu, Xun and Su, Liangjun},
  title   = {Determining the Number of Groups in Latent Panel Structures with an Application to Income Democracy},
  journal = {Quantitative Economics},
  year    = {2017}, volume = {8}, number = {3}, pages = {729--760}
}

@article{roth2023,
  author  = {Roth, Jonathan and Sant'Anna, Pedro H. C. and Bilinski, Alyssa and Poe, John},
  title   = {What's Trending in Difference-in-Differences? A Synthesis of the Recent Econometrics Literature},
  journal = {Journal of Econometrics},
  year    = {2023}, volume = {235}, number = {2}, pages = {2218--2244}
}

@article{rambachan2023,
  author  = {Rambachan, Ashesh and Roth, Jonathan},
  title   = {A More Credible Approach to Parallel Trends},
  journal = {The Review of Economic Studies},
  year    = {2023}, volume = {90}, number = {5}, pages = {2555--2591}
}

@article{barry1992,
  author  = {Barry, Daniel and Hartigan, John A.},
  title   = {Product Partition Models for Change Point Problems},
  journal = {The Annals of Statistics},
  year    = {1992}, volume = {20}, number = {1}, pages = {260--279}
}

@article{shen2010,
  author  = {Shen, Xiaotong and Huang, Hsin-Cheng},
  title   = {Grouping Pursuit Through a Regularization Solution Surface},
  journal = {Journal of the American Statistical Association},
  year    = {2010}, volume = {105}, number = {490}, pages = {727--739}
}

@article{hill2011,
  author  = {Hill, Jennifer L.},
  title   = {Bayesian Nonparametric Modeling for Causal Inference},
  journal = {Journal of Computational and Graphical Statistics},
  year    = {2011}, volume = {20}, number = {1}, pages = {217--240}
}

@article{fisher1958,
  author  = {Fisher, Walter D.},
  title   = {On Grouping for Maximum Homogeneity},
  journal = {Journal of the American Statistical Association},
  year    = {1958}, volume = {53}, number = {284}, pages = {789--798}
}

\newpage
\appendix

\section{Identification under Limited Overlap}
\label{app:identification}

In some staggered designs a cohort-time cell lacks a clean comparison group, for
example when every unit is eventually treated and the last-treated cells have no
not-yet-treated or never-treated controls. Such a cell $k$ then contributes no
independent identifying variation, its within-transformed column $\tilde{D}_k$ is
collinear with the others, $\Dtd'\Dtd$ is singular, and the flexible estimator
$\hat{\tau}_{k,\text{flex}}$ is undefined for that coordinate. This is the
staggered-DiD counterpart of the lack-of-overlap problem in cross-sectional
designs, where the fully interacted (``long'') regression is not well defined
\citep{kwon2026}. Partial homogeneity offers a route through it, by letting the
unidentified cell inherit its group's effect. Let
$\mathcal{U} \subseteq \{1,\ldots,K\}$ be the set of cells not identified by the
flexible model and $\mathcal{U}^c$ its complement.

\begin{proposition}[Identification via partial homogeneity]
  \label{prop:identification}
  Suppose the true partition is $\mathcal{P} = \{C_1,\ldots,C_m\}$ and each group
  contains at least one identified cell, $C_p \cap \mathcal{U}^c \neq \emptyset$
  for every $p$. Then under the restricted model \eqref{eq:within} the group
  effect $\phi_p$ is identified for every $p$, and hence $\tau_k = \phi_p$ is
  identified for every $k \in C_p$, including the unidentified cells
  $k \in C_p \cap \mathcal{U}$.
\end{proposition}

\begin{proof}
The group effect $\phi$ is identified in \eqref{eq:within} exactly when the
grouped design $\Dtd^* = [\Dtd_1,\ldots,\Dtd_m]$, with
$\Dtd_p = \sum_{k\in C_p}\tilde{D}_k$, has full column rank. Suppose
$\sum_p a_p \Dtd_p = 0$ and set $b_k = a_{p(k)}$, the coefficient of the group
containing $k$; then $\sum_k b_k \tilde{D}_k = 0$, so $b$ lies in the null space
of the flexible design $\tilde{D}$. A cell $k$ is identified
($k\in\mathcal{U}^c$) precisely when its column is not in the span of the others,
equivalently when every null-space vector of $\tilde{D}$ vanishes in coordinate
$k$. Each group $C_p$ contains such a cell $k_p^{*}\in C_p\cap\mathcal{U}^c$, so
$b_{k_p^{*}} = a_p = 0$; since this holds for every $p$, $a = 0$ and the group
columns are linearly independent. Hence $\Dtd^{*\prime}\Dtd^*$ is invertible,
each $\phi_p$ is estimable from \eqref{eq:oracle_var}, and assigning
$\tau_k = \phi_p$ identifies every cell in the group, including the unidentified
$k\in C_p\cap\mathcal{U}$.
\end{proof}

Two caveats bound the result's practical reach. First, the identifying content
comes entirely from the restriction, not the data. An unidentified cell carries
no independent variation, so which group it belongs to cannot be learned from the
outcomes, and Proposition~\ref{prop:identification} is an identification result
conditional on the grouping whose recovered effect is only as credible as the
(untestable) assignment. Second, the condition that each group contains an identified cell is exactly what
the result needs, and it is not tied to orthogonality: the argument uses only that
each group holds a cell outside the null space of $\tilde{D}$, so it covers
arbitrary collinear designs, and it is sharp, since a group made up entirely of
unidentified cells leaves its effect $\phi_p$ unidentified. Orthogonality only
simplifies the picture, an unidentified cell is then a zero column that contributes
nothing to its group, whereas under general collinearity the cell folds its
dependent variation into the group regressor and $\phi_p$ is identified but defined
to include it. The Bayesian estimator's co-clustering probabilities
describe this assignment, but for a genuinely unidentified cell they are governed
by the prior rather than by evidence, so they should be read as prior belief
rather than as data.

\section{Computing the $\ell_0$-PH Estimator}
\label{app:algorithm}

The $\ell_0$-PH estimator is computed in two steps, an agglomerative search (a greedy sequence of pairwise merges) for the partition and an ordinary least squares re-estimation of the grouped model under it. This appendix gives both and shows why the re-estimation is necessary.

  \paragraph{Step 1: Partition recovery.} Exact minimization of $Q(\mathcal{P})$ over all partitions would require searching the
  partitions of $\{1,\dots,K\}$, whose number is the Bell number $B_K$ and grows faster than exponentially, so exhaustive
  evaluation is infeasible beyond small $K$; we instead use an agglomerative algorithm that runs in $O(K^3)$ time, in the spirit
  of classical optimal-grouping procedures \citep{fisher1958}. Initialise with $K$ singleton groups, each assigned the flexible
  CATT estimate $\hat{\tau}_{k,\mathrm{flex}}$ and effective sample size $n_k = \|\tilde{D}_k\|^2$. At each iteration, compute
  for every pair of current groups $(A, B)$
  \begin{equation}
    \Delta\mathrm{Obj}(A, B) \;=\; \frac{n_A\, n_B}{n_A + n_B}\bigl(\hat{\tau}_A - \hat{\tau}_B\bigr)^2 \;-\; \lambda \cdot |A|
  \cdot |B|,
    \label{eq:deltaobj}
  \end{equation}
  where $n_A$, $n_B$ and $\hat{\tau}_A$, $\hat{\tau}_B$ are the current effective sample sizes and pooled estimates for groups
  $A$ and $B$ respectively. If $\min_{A \neq B}\,\Delta\mathrm{Obj}(A,B) \geq 0$, terminate and return the current partition
  $\hat{\mathcal{P}}$. Otherwise merge the minimizing pair, updating
  \begin{equation}
    \hat{\tau}_{A \cup B} \;=\; \frac{n_A\,\hat{\tau}_A + n_B\,\hat{\tau}_B}{n_A + n_B}, \qquad n_{A \cup B} \;=\; n_A + n_B.
    \label{eq:pooled}
  \end{equation}
  The algorithm terminates after at most $K - 1$ merges.

  \paragraph{Step 2: Re-estimation under the discovered partition.} Let $\hat{\mathcal{P}} = \{\hat{C}_1, \ldots,
  \hat{C}_{\hat{m}}\}$ denote the partition returned by Step~1. For each group $\hat{C}_p$, define the group regressor
  $\tilde{D}_p = \sum_{k \in \hat{C}_p} \tilde{D}_k$ and collect the $\hat{m}$ group regressors into the restricted
  within-transformed design matrix $\tilde{D}^* = [\tilde{D}_1, \ldots, \tilde{D}_{\hat{m}}]$. Estimate the restricted model
  \begin{equation}
    \tilde{Y} \;=\; \tilde{D}^*\phi + \tilde{\varepsilon}, \qquad \tilde{\varepsilon} \sim \bigl(0,\, \sigma^2 I_r\bigr),
    \label{eq:step2}
  \end{equation}
  by OLS to obtain $\hat{\phi} = (\tilde{D}^{*\prime}\tilde{D}^*)^{-1}\tilde{D}^{*\prime}\tilde{Y}$. The $\ell_0$-PH CATT estimate
  for cell $k$ is $\hat{\tau}_k^{\ell_0} = \hat{\phi}_p$ for all $k \in \hat{C}_p$. Equivalently, \eqref{eq:step2} corresponds to
   the TWFE regression
  \begin{equation}
    Y_{igt} \;=\; \alpha_i + \gamma_t + \sum_{p=1}^{\hat{m}} \phi_p \sum_{k \in \hat{C}_p} D_{k,igt} + \varepsilon_{it},
    \label{eq:step2twfe}
  \end{equation}
  estimated on the full panel. We refer to the complete two-step procedure as the $\ell_0$-PH estimator.

  \paragraph{Why re-estimation is essential.} A natural one-step alternative uses the greedy-pooled values \eqref{eq:pooled}
  directly as CATT estimates. We show that this approach is generically inefficient and that the OLS re-estimation in Step~2 is a
   necessary correction, not a refinement.

  The greedy update $\hat{\tau}_{A \cup B}$ in \eqref{eq:pooled} is the weighted mean of the flexible estimates with weights $n_k
   = \|\tilde{D}_k\|^2$. Under the orthogonality condition $\tilde{D}_j'\tilde{D}_k = 0$ for $j \neq k$, which holds approximately in
  balanced panels in which each unit belongs to exactly one cohort, $\tilde{D}^{*\prime}\tilde{D}^*$ is diagonal with $p$-th
  diagonal entry $\|\tilde{D}_p\|^2 = \sum_{k \in \hat{C}_p} n_k$. In this case, the OLS estimate from Step~2 reduces to
  \begin{equation}
    \hat{\phi}_p \;=\; \frac{\tilde{D}_p'\tilde{Y}}{\|\tilde{D}_p\|^2} \;=\; \frac{\displaystyle\sum_{k \in \hat{C}_p}
  n_k\,\hat{\tau}_{k,\mathrm{flex}}}{\displaystyle\sum_{k \in \hat{C}_p} n_k},
    \label{eq:phorthog}
  \end{equation}
  which coincides with the greedy-pooled value \eqref{eq:pooled}. By Proposition~\ref{prop:oracle_dom}, $\hat{\phi}_p$ is therefore the BLUE of the
  group treatment effect, and Steps~1 and 2 deliver the same answer when the design is orthogonal.

  In general panels, however, cells belonging to the same treatment cohort share unit fixed effects, and cells falling on the
  same calendar period share time fixed effects. Both sources of co-occurrence produce non-zero off-diagonal elements in
  $\tilde{D}^{*\prime}\tilde{D}^*$, so the OLS estimate from Step~2,
  \begin{equation}
    \hat{\phi} \;=\; \bigl(\tilde{D}^{*\prime}\tilde{D}^*\bigr)^{-1}\tilde{D}^{*\prime}\tilde{Y},
    \label{eq:phgeneral}
  \end{equation}
  no longer reduces to the weighted mean \eqref{eq:pooled}. Three consequences follow. First, the greedy-pooled estimator does
  not satisfy the normal equations of \eqref{eq:step2} and therefore fails to be BLUE. Second, the $\Delta\mathrm{Obj}$ criterion
   \eqref{eq:deltaobj} is an approximation, in that it treats $\tilde{D}^{*\prime}\tilde{D}^*$ as diagonal when computing the RSS cost of
   each candidate merge, so partition selection in Step~1 rests on an approximate objective. Third, in a non-orthogonal design the greedy-pooled mean is
  inefficient rather than inconsistent: under a correct merge $\tau_A^* = \tau_B^* = \phi_p$, so its probability limit
  $(n_A\tau_A^* + n_B\tau_B^*)/(n_A + n_B)$ equals $\phi_p$ whatever the design, but it is not the cross-product-weighted
  combination that \eqref{eq:step2} solves for and so is not efficient. Inconsistency for the cell effects arises only when cells
  with distinct true effects are merged, a Step-1 selection error, not from non-orthogonality.

  Step~2 corrects the two estimation shortcomings exactly and for any design. Given the partition $\hat{\mathcal{P}}$ returned by
  Step~1, $\hat{\phi}$ in \eqref{eq:phgeneral} solves the exact grouped normal equations, so it is the best linear unbiased
  estimator of the group effects under $\hat{\mathcal{P}}$ whether or not the design is orthogonal, with variance
  $\sigma^2[(\tilde{D}^{*\prime}\tilde{D}^*)^{-1}]_{pp}$; under orthogonality this is $\sigma^2/\sum_{k \in \hat{C}_p} n_k$,
  strictly below $\sigma^2/n_k$ for every $k \in \hat{C}_p$ with $|\hat{C}_p| \geq 2$, the ratio $n_k/\sum_{j \in \hat{C}_p} n_j <
  1$ of \eqref{eq:var_ratio}, and the variance reduction grows with the group size. The one-step greedy-pooled estimator attains
  this fit only under exact orthogonality; Step~2 attains it for any design. The remaining shortcoming is of a different kind: it
  concerns which merges Step~1 selects, and Step~2 re-estimates within $\hat{\mathcal{P}}$ rather than reopening the approximate
  search, so it takes the partition as given. The unbiased efficiency dominance of Proposition~\ref{prop:oracle_dom} is therefore
  inherited only when $\hat{\mathcal{P}}$ is the true partition; a false merge trades the variance gain for a specification bias
  that re-estimation cannot remove (Remark~\ref{rem:postselection}), the low-separation behavior documented in
  Section~\ref{sec:sim}.

  The Bayesian estimator underscores the same point. Given a partition, its posterior mean of $\phi$ in \eqref{eq:bayes:phipost} solves the same grouped normal equations as \eqref{eq:step2} and, in the diffuse limit $\sigma_0^2 \to \infty$, equals $\hat{\phi}$ in \eqref{eq:phgeneral}. Both target the PH estimator under the discovered partition, whereas the one-step greedy-pooled mean does not.

\section{Heterogeneous Errors}
\label{app:hetero}

The model of Section~\ref{sec:bayes} assumes a scalar error variance, $\tilde\varepsilon \mid \tilde D \sim N(0,\sigma^2 I_r)$ in \eqref{eq:bayes:lik}. This appendix relaxes that assumption to a heterogeneous error covariance and records the few quantities that change, together with the augmented Gibbs sampler. Neither the partition prior \eqref{eq:bayes:partprior} nor the sampler's cluster-assignment logic is affected. The only changes are the design moments that enter the marginal likelihood and the group-effect posterior, and one additional Gibbs block that updates the variance parameters.

\paragraph{The modified model.} Replace the scalar covariance in \eqref{eq:bayes:lik} with a structured positive-definite matrix on the transformed model,
\begin{equation}
  \tilde\varepsilon \mid \tilde D \sim N(0,\Omega), \qquad \Omega = \bigoplus_{j=1}^{J}\sigma_j^2 I_{n_j},
  \label{eq:het:cov}
\end{equation}
where the within-transformed observations are partitioned into $J$ variance blocks $\mathcal B_1,\dots,\mathcal B_J$ of sizes $n_1,\dots,n_J$, indexed by a known label $v(i)\in\{1,\dots,J\}$. The blocks collect observations expected to share a noise level, for instance by cohort, by treated versus comparison status, or by an observed cluster. The homoskedastic model is the special case $J=1$. A fully general non-diagonal $\Omega$, such as the clustered influence-function covariance already used in the empirical applications (Remark~\ref{rem:sigma}), enters the formulas below in the same way through $\Omega^{-1}$; a covariance that is block-diagonal in the original observations residualizes to such a dense $\Omega$, so the per-block update \eqref{eq:het:sigma} is specific to the block-diagonal working parameterization, while the marginal likelihood and group-effect posterior below hold for any $\Omega$. We complete the model with an independent conjugate prior on each variance level,
\begin{equation}
  \sigma_j^2 \sim \mathrm{Inv\text{-}Gamma}(a_0,b_0), \qquad j=1,\dots,J,
  \label{eq:het:prior}
\end{equation}
the natural generalization of \eqref{eq:bayes:sigma}.

\paragraph{What changes in the theory.} Every appearance of the precision $\sigma^{-2}I_r$ becomes $\Omega^{-1}$. Writing the two design moments in the $\Omega^{-1}$ metric,
\begin{equation}
  S \;=\; \tilde D^{*\prime}\Omega^{-1}\tilde D^{*}, \qquad u \;=\; \tilde D^{*\prime}\Omega^{-1}\tilde Y,
  \label{eq:het:moments}
\end{equation}
the group effects still integrate out in closed form. Conditional on $\Omega$, the marginal likelihood \eqref{eq:bayes:marg} becomes $\tilde y\mid\mathcal P \sim N(\mu_0\tilde D^{*}\mathbf 1_m,\ \Omega+\sigma_0^2\tilde D^{*}\tilde D^{*\prime})$, and the Woodbury reduction \eqref{eq:bayes:logmarg}, dropping terms constant in $\mathcal P$, reads
\begin{equation}
  \log p(\tilde y\mid\mathcal P,\Omega) \;\propto\; -\tfrac12\log\bigl|\,I_m+\sigma_0^2 S\,\bigr| \;+\; \tfrac{\sigma_0^2}{2}\,u_0'\bigl(I_m+\sigma_0^2 S\bigr)^{-1}u_0, \qquad u_0 \;=\; u - \mu_0\, S\mathbf 1_m.
  \label{eq:het:logmarg}
\end{equation}
The group-effect posterior \eqref{eq:bayes:phipost} becomes $\phi\mid\mathcal P,\Omega \sim N(\mu^*,\Sigma^*)$ with
\begin{equation}
  \Sigma^* \;=\; \bigl(S + I_m/\sigma_0^2\bigr)^{-1}, \qquad \mu^* \;=\; \Sigma^*\bigl(u + \mu_0\mathbf 1_m/\sigma_0^2\bigr),
  \label{eq:het:phipost}
\end{equation}
which is the ridge-regularized GLS estimator of the group effects. Setting $\Omega=\sigma^2 I_r$ gives $S=\tilde D^{*\prime}\tilde D^{*}/\sigma^2$ and $u=\tilde D^{*\prime}\tilde Y/\sigma^2$ and recovers \eqref{eq:bayes:logmarg} and \eqref{eq:bayes:phipost} exactly. The computational cost is unchanged, the $m\times m$ matrix $S$ still governs the calculation, and forming it costs one weighting of the design by $\Omega^{-1}$, which is $O(NT)$ for the diagonal case \eqref{eq:het:cov}.

The $\ell_0$ correspondence of Section~\ref{sec:bayes:sigmafixed} survives as long as the covariance is held fixed. With $\Omega$ known, the profile over the group effects replaces the ordinary residual sum of squares by its $\Omega^{-1}$-weighted (GLS) counterpart $(\tilde Y-\Dtd^*\hat\phi)'\Omega^{-1}(\tilde Y-\Dtd^*\hat\phi)$, while $\log|\Omega|$ is constant across partitions and drops out, so the MAP partition still minimizes a residual sum of squares plus a fixed penalty of $2\lambda_0$ per cross-group pair, now with the GLS fit in place of the least-squares fit. This is the relevant form when a plug-in covariance is used, as in the applications, where $\Omega$ is the fixed first-stage $\hat\Sigma$ of Lemma~\ref{lem:sufficiency}. What does not carry over is the case of an \emph{unknown} $\Omega$: once the variance levels are estimated, profiled, or integrated out, they enter through $\log|\Omega|$ and through their dependence on the fitted residuals, and the objective is no longer a fixed penalty added to a residual sum of squares. Heterogeneous errors with unknown variances are therefore handled inside the sampler, which averages over the variance levels along with the partition.

\paragraph{The augmented Gibbs sampler.} The collapsed sampler of Section~\ref{sec:bayes:gibbs} gains one block. Given a current $\Omega$, one sweep runs the following three steps.
\begin{enumerate}
\item For each CATT $k$, remove $k$ from its cluster, delete the cluster if it is empty, and draw $z_k$ from the conditional probabilities \eqref{eq:bayes:condprob} with the marginal likelihoods now evaluated from \eqref{eq:het:logmarg}.
\item Draw the group effects $\phi \sim N(\mu^*,\Sigma^*)$ from \eqref{eq:het:phipost}.
\item Form the within-transformed residuals $\hat{\tilde\varepsilon} = \tilde Y - \tilde D^{*}\phi$ under the current partition, and for each variance block draw
\begin{equation}
  \sigma_j^2 \mid \phi,\mathcal P,\tilde Y \;\sim\; \mathrm{Inv\text{-}Gamma}\!\Bigl(a_0+\tfrac{n_j}{2},\ b_0+\tfrac12\!\!\sum_{i\,:\,v(i)=j}\!\!\hat{\tilde\varepsilon}_i^{\,2}\Bigr),
  \label{eq:het:sigma}
\end{equation}
then assemble $\Omega=\bigoplus_j\sigma_j^2 I_{n_j}$ for the next sweep.
\end{enumerate}
Reinstating $\phi$ in Step~2 before the variance draw is what makes Step~3 conjugate, exactly as the original sampler reinstates $\phi$ to report the group effects. All posterior summaries of Section~\ref{sec:bayes:gibbs}, the CATT posterior mean \eqref{eq:bayes:tauhat}, the co-clustering matrix \eqref{eq:bayes:cocluster}, and the aggregated estimands, are formed from the augmented draws $\{z^{(s)},\phi^{(s)},\Omega^{(s)}\}$ without further change, so the reported credible intervals now propagate uncertainty in the error variances alongside the partition. When a plug-in covariance is preferred, fixing $\Omega$ at an estimate $\hat\Omega$ and omitting Step~3 recovers a single GLS-posterior sampler, the direct analog of the fixed-$\hat\sigma^2$ special case of Section~\ref{sec:bayes:sigmafixed}.

\section{Fixed versus Random $\sigma^{2}$}
\label{app:sigma}

The main text draws $\sigma^{2}$ from its inverse-gamma full conditional each sweep, the full Bayesian model of Section~\ref{sec:bayes:model}. Holding $\sigma^{2}$ fixed at the plug-in $\hat\sigma^{2}$ is the special case that yields the exact $\ell_0$ correspondence (Section~\ref{sec:bayes:sigmafixed}). This appendix shows the two give the same inference. We rerun the coverage experiment of Section~\ref{sec:sim:cov} at $m^{\ast}=6$ and $\delta=6$ over $500$ replications, running two separate collapsed chains on each replication, one that draws $\sigma^{2}$ from its full conditional each sweep and one that holds it at the plug-in $\hat\sigma^{2}$. Each evaluates the cluster-assignment marginal likelihood \eqref{eq:bayes:logmarg} at its own $\sigma^{2}$, so the two partition posteriors are targeted separately, and the credible intervals of Table~\ref{tab:sigma} come from the matching chain.

\begin{table}[H]
  \centering
  \caption{Coverage and average length of nominal $95\%$ DP credible intervals
  under a random (full Bayesian) and a fixed (plug-in) error variance,
  $m^{\ast}=6$, $\delta=6$, $500$ replications.}
  \label{tab:sigma}
  \small
  \begin{tabular}{lcccc}
    \toprule
    & \multicolumn{2}{c}{CATT} & \multicolumn{2}{c}{ATT} \\
    \cmidrule(lr){2-3}\cmidrule(lr){4-5}
    $\sigma^{2}$ treatment & Coverage & Length & Coverage & Length \\
    \midrule
    Random (full Bayesian) & 0.93 & 0.20 & 0.92 & 0.11 \\
    Fixed (plug-in)        & 0.93 & 0.20 & 0.93 & 0.11 \\
    \bottomrule
  \end{tabular}
\end{table}

Table~\ref{tab:sigma} reports the result. The two samplers deliver essentially the same coverage and interval length, for the cohort-time effects and for the overall ATT. The reason is the one given in Section~\ref{sec:bayes:sigmafixed}, the residual degrees of freedom $NT - N - T + 1 - K$ are large, so $\sigma^{2}$ is pinned down and the estimation error a prior on it would propagate is negligible. The choice is therefore immaterial for the reported results, and we use the fixed-$\sigma^{2}$ case in the main text only to make the $\ell_0$ connection concrete.

\section{MCMC Diagnostics}
\label{app:mcmc}

This appendix documents that the reported credible intervals and co-clustering
probabilities reflect the stationary posterior rather than transient short-chain
behavior. We validate the collapsed Gibbs sampler in two ways: against an exact
enumeration of the posterior in the small-$K$ application, and through standard
multi-chain convergence diagnostics in the simulation.

\paragraph{Exact-enumeration benchmark ($K=7$).} In the \citet{callaway2021}
application the seven cohort-time cells admit only $B_7 = 877$ partitions, so the
exact partition posterior can be computed in closed form: for each partition we
evaluate the collapsed marginal likelihood under the exact within-cell covariance
and the CRP prior, and normalize over all $877$. Table~\ref{tab:mcmc} compares the
sampler (one chain of $20{,}000$ sweeps after $2{,}000$ burn-in) against this exact
posterior. The two agree to within Monte Carlo error: the posterior expected
number of groups matches ($2.20$), the distribution of the group count agrees to
within $0.004$ at every value, the overall-ATT posterior mean and $95\%$ interval
coincide, and the largest discrepancy over all $21$ pairwise co-clustering
probabilities is $0.012$. The exact within-cell covariance is what makes this
agreement possible: a diagonal (independence) approximation to the cluster-assignment
moves reproduces the same aggregates but misstates individual co-clustering
probabilities by up to $0.39$ in this correlated design, so we use the exact
covariance throughout (Remark~\ref{rem:sigma}). Figure~\ref{fig:mcmc} (left) plots
the sampler's co-clustering probabilities against the exact values.

\begin{table}[H]
  \centering
  \caption{Exact-enumeration benchmark for the \citet{callaway2021} application
  ($K=7$, $B_7=877$ partitions). The collapsed Gibbs sampler against the exact
  partition posterior computed in closed form.}
  \label{tab:mcmc}
  \small
  \begin{tabular}{lcc}
    \toprule
    Quantity & Gibbs sampler & Exact enumeration \\
    \midrule
    Posterior E[\# groups]                 & $2.20$ & $2.20$ \\
    Overall-ATT posterior mean             & $-0.024$ & $-0.024$ \\
    Overall-ATT $95\%$ interval            & $[-0.047,\,-0.000]$ & $[-0.048,\,-0.000]$ \\
    Max pairwise co-clustering gap vs.\ exact & $0.012$ & --- \\
    \bottomrule
  \end{tabular}
\end{table}

\paragraph{Multi-chain diagnostics ($K=18$).} In the simulation, where $K=18$
makes enumeration infeasible, we run four chains from dispersed initializations,
all singletons, one pooled group, and two random partitions into nine and three
groups, at $m^{\ast}=6$, $\delta=6$. The Gelman--Rubin statistic is $\hat R = 1.00$
for both the overall ATT and the number of groups, and the effective sample size
across the four chains is about $11{,}400$ for the ATT and $8{,}600$ for the group
count. Figure~\ref{fig:mcmc} (right) shows the cluster-count traces converging to a
common region within the burn-in from all four starts. As a check on chain length,
the $100$-draw setting used in the main experiments returns an overall-ATT interval
close to that of a $3{,}000$-draw chain, with a stable center and an upper endpoint
that moves only from about $0.078$ to $0.085$, confirming that the reported
quantities are not artifacts of chain length.

\begin{figure}[H]
  \centering
  \begin{subfigure}{0.49\linewidth}\includegraphics[width=\linewidth]{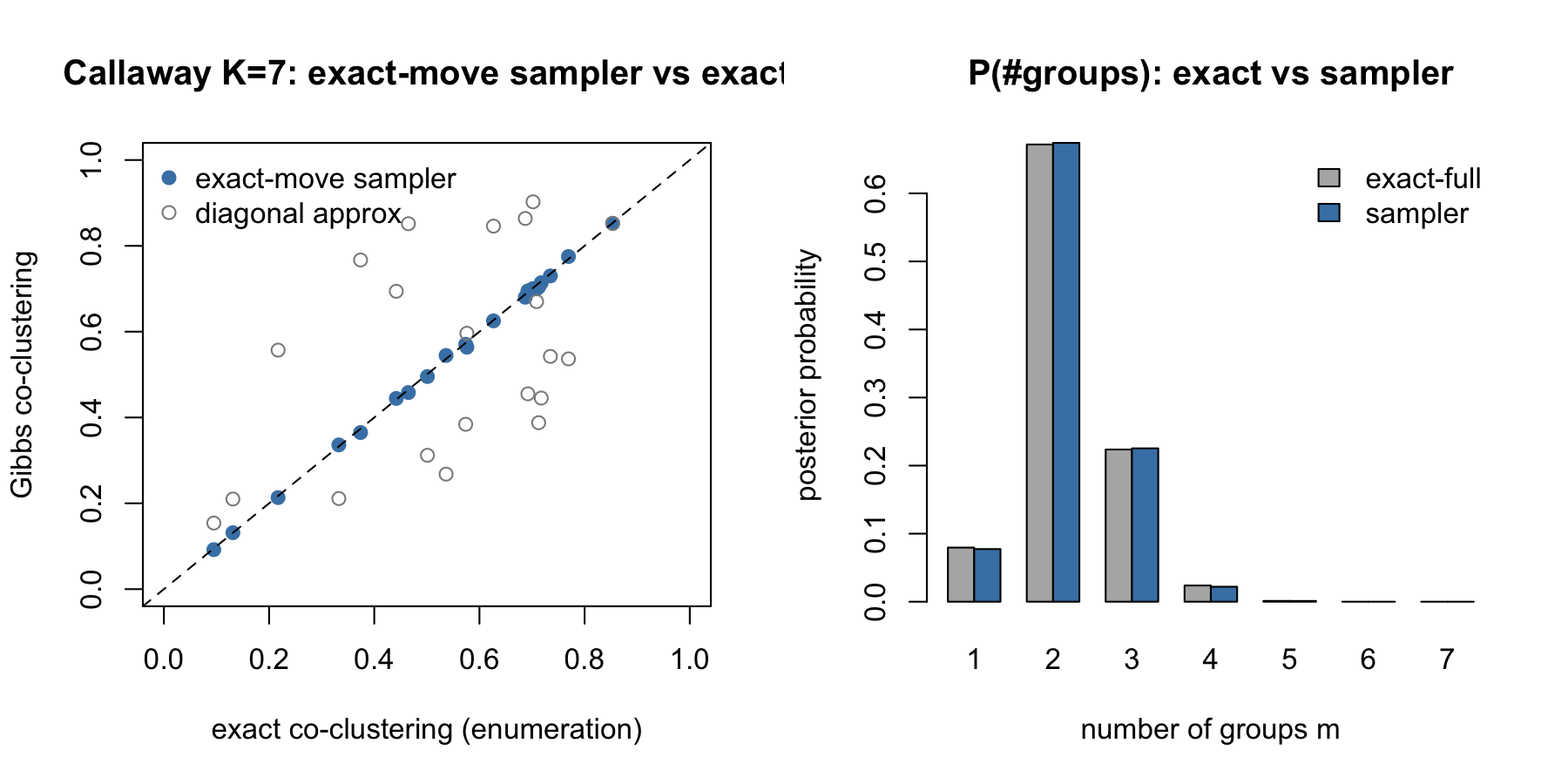}\end{subfigure}
  \hfill
  \begin{subfigure}{0.49\linewidth}\includegraphics[width=\linewidth]{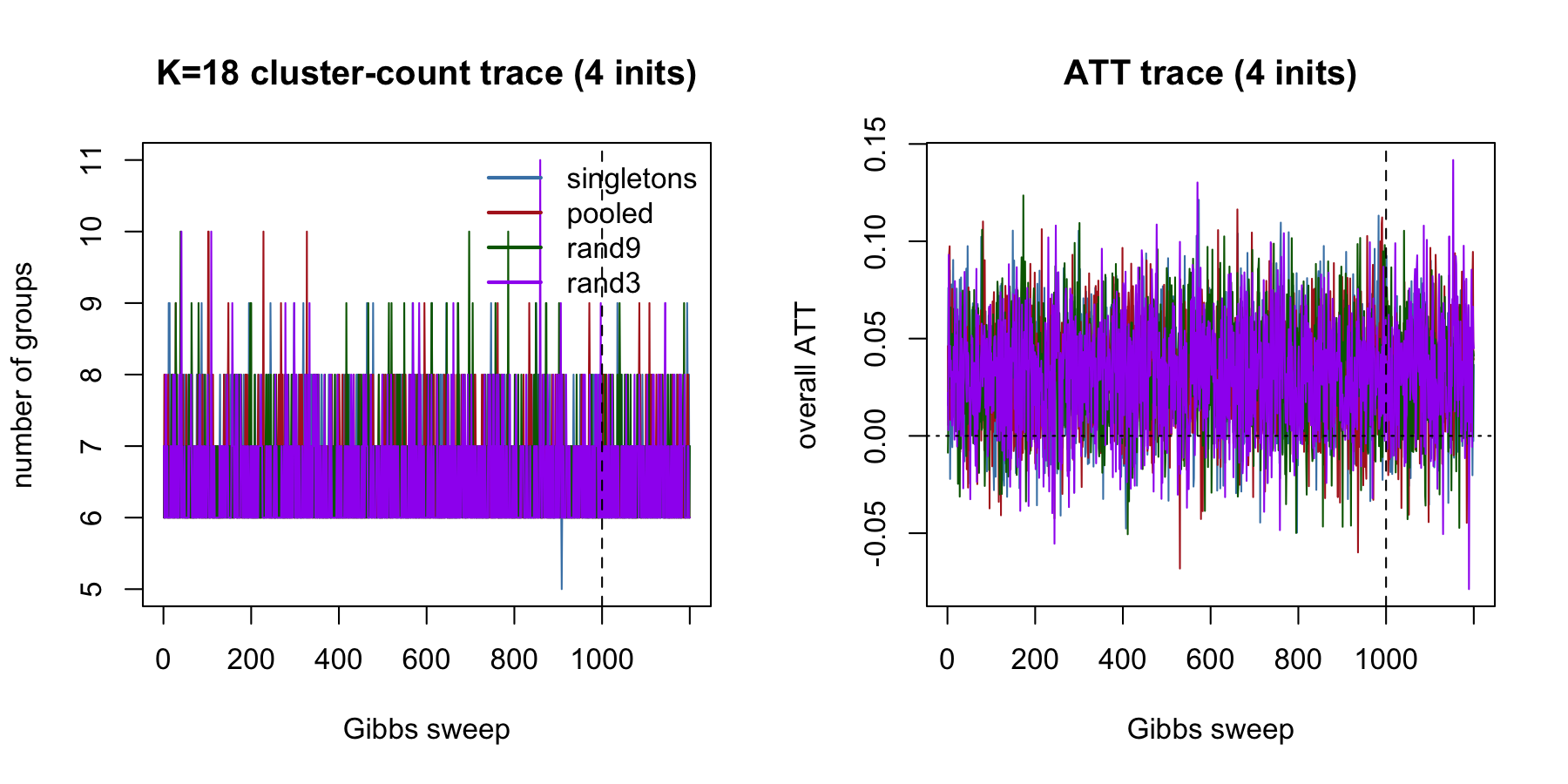}\end{subfigure}
  \caption{Left: Callaway application ($K=7$), Gibbs co-clustering probabilities
  against the exact enumeration (filled points on the $45^{\circ}$ line); the
  diagonal approximation (open points) is shown for contrast. Right: simulation
  ($K=18$), cluster-count traces from four dispersed initializations converging
  within the burn-in (dashed line).}
  \label{fig:mcmc}
\end{figure}

\end{document}